\documentclass[12pt]{article}
\usepackage{siunitx}
\usepackage{geometry}
\usepackage{graphicx}
\usepackage{tikz} 
\usepackage{titlesec}
\usepackage{lipsum}
\usepackage{amsmath}
\usepackage{qtree}
\usepackage{multirow}
\usepackage{float}
\usepackage{amsthm}
\usepackage{adjustbox}
\usepackage{amssymb}
\usepackage{amstext}
\usepackage{enumerate}
\usepackage{paralist}
\usepackage{caption}
\usepackage{hyperref}
\usepackage{url}
\usepackage{subcaption}
\usepackage{makecell}
\usepackage{soul}
\usepackage{tabu}
\usepackage{colortbl}
\usepackage[ruled, vlined]{algorithm2e}
\usepackage{algpseudocode}
\usepackage{verbatim}
\usepackage{authblk}
\usepackage[style=numeric, sorting = none]{biblatex}
\newtheorem{proposition}{proposition}
\DeclareMathOperator*{\argmin}{arg\,min}
\usepackage{subfiles} 
\DeclareMathOperator*{\argmaxA}{arg\,max}

\title{IIVA: A Simulation Based Generalized Framework for \underline{\textbf{I}}nterdependent \underline{\textbf{I}}nfrastructure \underline{\textbf{V}}ulnerability \underline{\textbf{A}}ssessment}
\author{Prasangsha Ganguly \& Sayanti Mukherjee*}
\affil[]{Department of Industrial and Systems Engineering, University at Buffalo}
\affil[]{*sayantim@buffalo.edu}

\date{}
\begin{document}

\maketitle

\begin{abstract}
Accurate vulnerability assessment of critical infrastructure systems is cardinal to enhance infrastructure resilience. Unlike traditional approaches, this paper proposes a novel infrastructure vulnerability assessment framework that accounts for: various types of infrastructure interdependencies including physical, logical and geographical from a holistic perspective; lack of/incomplete information on supply-demand flow characteristics of interdependent infrastructure; and, unavailability/inadequate data on infrastructure network topology and/or interdependencies. Specifically, this paper models multi-infrastructure vulnerabilities leveraging simulation-based hybrid approach coupled with time-dependent Bayesian network analysis while considering cascading failures within and across CIS networks, under incomplete information. Existing synthetic data on electricity, water and supply chain networks are used to implement/validate the framework. Infrastructure vulnerabilities are depicted on a geo-map using Voronoi polygons. Our results indicate that infrastructure vulnerability is inversely proportional to the number of redundancies inbuilt in the infrastructure system, indicating that allocating resources to add redundancies in an existing infrastructure system is essential to reduce its risk of failure. It is observed that higher the initial failure rate of the components, higher is the vulnerability of the infrastructure, highlighting the importance of modernizing and upgrading the infrastructure system aiming to reduce the initial failure probabilities. Our results also underline the importance of collaborative working and sharing the necessary information among multiple infrastructure systems, aiming towards minimizing the overall failure risk of interdependent infrastructure systems. The proposed scenario-based simulation illustrates regional variations in various vulnerability estimates under changing parameters/scenarios, offering infrastructure managers various customizable decision-parameters for efficient resource allocations. 
\end{abstract}



\section{INTRODUCTION}
The modern society is critically dependent on lifeline infrastructure systems such as electric power, telecommunication, water, transportation, and others. Most of these infrastructure systems comprise of multiple components interconnected by  logical or physical dependencies, rendering them suitable to be modeled as a network of networks where multiple infrastructure networks are connected by interdependency links. In a network-of-network system, failure of a key component may lead to a cascading failure of the entire system, which is not desirable~\cite{Faturechi, Dong}. 
In literature, modeling the reliability of an infrastructure is well studied~\cite{RE_book1, RE_book2}. One of the most popular methods used to model system reliability is the Fault Tree (FT) method~\cite{Dugan_dft}. In a FT Analysis, a top event (failure of the whole system) is represented as a Boolean combination (AND, OR, NOT) of the basic events (e.g., failure of switch, substation(s), transformer(s)). In this method, at first the failure probabilities are assigned to the basic events and then the failure probability of the top event is calculated through Boolean modeling and calculations of probabilities~\cite{FT}. The main \textit{advantage} of such deterministic models is that, they are easily interpretable and encourages a methodological way of thinking. However, the main \textit{disadvantage} is that, being static models, they can't capture the dynamics of failure propagation with time. Several research studies have extended the static fault trees into dynamic fault trees (DFT) that uses different advanced gates (e.g., the spare gates, priority gates etc.) to capture the temporal dynamics of the failure propagation~\cite{Dugan_dft}. Time dependent Bayesian networks such as the discrete time Bayesian networks, the continuous time Bayesian networks, and the Dynamic Bayesian networks (DBN) have been used to model the DFT in literature~\cite{Raiteri, Boudali2}. However, the methods for developing a FT or DFT are not generalized enough to be applied to different infrastructure systems. Even for a single infrastructure system like electricity infrastructure, different FTs are proposed based on the components or events considered such as the failure of generators, switches, substations, etc.~\cite{VOLKANOVSKI} or disruption of the distribution generators, output flow path, etc.~\cite{Song}. Nonetheless, a generalized approach to construct a DFT would be particularly useful for the infrastructure managers. Even if such DFTs may not be exact, a procedural approach to construct a heuristic DFT that can be applied to multiple infrastructure systems will be helpful for the reliability assessment of any type of infrastructure. 

Since the infrastructure systems are interdependent, analyzing a single infrastructure system is not ideal as it underestimates the risk of failure. Hence, recently several studies have focused on modeling the vulnerability and reliability of multiple interdependent infrastructure systems in conjunction~\cite{Sharma, Lu, Galbusera}. Rinaldi \textit{et al.} classified infrastructure interdependencies into four categories:
\begin{inparaenum}[(i)]
\item \textit{physical interdependency}, which emerges from physical linkages or connections among elements of the 
infrastructures;
\item \textit{cyber interdependency}, when the state of one infrastructure depends on information/data transmitted from/through another infrastructure; 
\item \textit{geographical interdependency}, if there exists a close spatial proximity between elements of different infrastructures; and
\item \textit{logical interdependency}, when the state of one infrastructure depends on the state of others via 
mechanisms that are not physical, cyber, or geographic connections
\end{inparaenum}~\cite{Rinaldi}. Different other types of interdependencies such as functional interdependency have also been studied in the literature~\cite{Zhang2, Wallace}. Ouyang \textit{et al.} provided a thorough survey of the literature focusing on the interdependent infrastructure systems modeling~\cite{Ouyang_review}. Previous studies also established that for accurate estimation of overall system performance, the connections between different networks need to be considered by analyzing the overall system and individual network performance~\cite{Krishna, Hernandez, Yagan}. Considering the methodology of modeling the critical infrastructure systems, the existing literature can be broadly classified into six categories: 
\begin{inparaenum}[(i)]
\item \textit{aggregate supply-demand methodology}, which considers the supply-demand relationships by evaluating the total required demand for infrastructure services in a region, and the ability of satisfying that demand by the infrastructure system~\cite{Apostolakis, Adachi, Lee};
\item \textit{dynamic simulations}, which employ simulation techniques like discrete event simulation or system dynamics techniques~\cite{Conrad, zio, Osorio_simulation}; 
\item \textit{agent-based models}, which considers the physical components of an infrastructure to be modeled as agents and allow them to interact for the analysis of the operational characteristics and physical states of infrastructures~\cite{Casalicchio};
\item \textit{physics-based models}, which considers the physical aspects of an infrastructure using standard engineering techniques~\cite{Chen};
\item \textit{population mobility models}, which captures the movement of entities through geographical regions considering a very high resolution of modeling approach~\cite{Casa2};
\item \textit{Leontief input-output models}, which considers the economic flows among infrastructure sectors to provide a linear, aggregated, time-independent analysis of the generation, flow, and consumption of various commodities in the various sectors~\cite{Haimes1} \end{inparaenum}.

A considerable fraction of literature that aims to model the interconnected infrastructure systems use the supply-demand modeling approaches by explicitly identifying the interconnections or interdependencies between different infrastructure systems~\cite{Cavallaro, Cavdaroglu}. For example, Lee \textit{et al.} proposed a detailed mathematical programming formulation for network flow based model that explicitly incorporates the interdependencies among a set of civil infrastructure systems. However, this model requires a large set of parameters to be known in advance for the model to run~\cite{Lee}. Mixed integer linear programs (MILP)~\cite{Ouyang_cacie, Chen_net}, multi-stage optimization models~\cite{Fng} and other optimization techniques~\cite{Alinizzi, Gonzalez2} have been widely used in literature for modeling a resilient system of interdependent infrastructure systems. In these models, all the connector variables between different infrastructure systems are required to be known in advance. Several other studies focused on mathematical programming in solving interdependent network design problem (INDP)~\cite{Gonzalez, Lind, NURRE}, network reconstruction, restoration and reliability modeling~\cite{Veremyev}. Such methods suffer from similar limitations where the entire functional dependency links between two infrastructure systems need to be known beforehand. Furthermore, all the models require medium or large amount of data; especially, the flow based models require large amounts of data for realistic modeling of the interdependent infrastructure systems~\cite{Ouyang_simulation}. However, one of the main challenges in the research of interdependent infrastructure systems is the unavailability of data~\cite{Rinaldi}. The dynamic simulation based approaches solve this problem by abstracting the physical details of the services provided by the infrastructures in almost all aspects---making the model simple and feasible to use~\cite{Osorio_simulation, Ouyang_simulation}. However, often such methods are inherently incapable of capturing most of the dimensions of physical, logical and geographical interdependencies between the infrastructure systems~\cite{zio}. Hence a hybrid approach is required that can model most of the interdependence dimensions of the infrastructure systems, while considering that a minimal number of parameters is known beforehand. 

Therefore, although the interdependent infrastructure reliability and vulnerability assessment is well studied in literature, there exists notable research gaps. 
Hence, in this paper, a generalized framework \textit{to model the vulnerability of an infrastructure system arising from both the cascading failure within the network and the failure of its components induced by other interdependent infrastructure component failures} is proposed. Specifically, considering the vulnerability estimation problem from the perspective of the manager of a child infrastructure that is dependent on other parent infrastructures, we aim to address the following research gaps, identified through our comprehensive literature review, in a systematic manner as listed below. Note, in this paper, when an infrastructure provides service to another infrastructure, the \textit{service-providing infrastructure} is termed as a \textit{parent infrastructure}, and the \textit{service-receiving infrastructure} is termed as a \textit{child infrastructure}. 
\begin{itemize}
    \item \begin{inparaenum}[*]
    \item \textit{Gap: } There is a lack of a holistic framework to estimate the vulnerability of interdependent infrastructure systems. The literature either focus on a single infrastructure system vulnerability or only the vulnerability induced by the interdependent infrastructure systems. Not being able to model the infrastructure vulnerability from a holistic perspective often underestimates the compound risk of failure induced by multiple interdependent infrastructure systems failures.
    \item \textit{Contribution: } To address this gap, a simulation environment to calculate the time dependent vulnerability of infrastructure systems measured by the probability of failure of services, while considering both the within infrastructure cascading failure and the failure induced by other interdependent infrastructures' failures is proposed. This estimate of vulnerability which encapsulates both the within infrastructure cascade and failure induced by other infrastructure systems into a single metric, is referred to as \textit{comprehensive vulnerability} for the subsequent discussions in the paper.
    \end{inparaenum}
    
    \item \begin{inparaenum}[*]
    \item \textit{Gap: } To estimate the failure propagation within a network, DFTs are well studied. However, depending on the infrastructure under consideration, the component and connectivity of the DFTs can significantly vary. For several infrastructure systems, DFTs do not even exist in literature. Furthermore, complex DFTs may suffer from state space explosion problem. Lack of a generalized approach for constructing DFTs for all types of infrastructure systems may hinder efficient assessment of reliability and vulnerability of the infrastructure systems.
    \item \textit{Contribution: } Therefore, to address this gap, a heuristic algorithm to construct a DFT of an infrastructure, which is generalized enough to be applied to multiple infrastructure systems is presented. Although the DFTs developed using the proposed heuristic approach is not exact, in absence of data this method can provide a close approximation of the exact DFT, and thus can aid in efficient vulnerability and reliability assessment of the infrastructure system under incomplete information.  
    \end{inparaenum}
    
    \item \begin{inparaenum}[*]
    \item \textit{Gap: } In contrary to the real life scenario, often the most realistic models assume that all the supply demand characteristics are well-known to all the infrastructure systems' managers in advance. In doing so, such models relax the intrinsic dimensions of the interdependencies while trying to abstract the supply demand characteristics. Not considering all the dimensions of interdependencies may underestimate the actual risk of failure of the infrastructure systems.
    \item \textit{Contribution: } To address this problem, a generalized framework that can capture most of the dimensions of the interdependencies, while requiring minimal data is proposed. Our proposed approach only assumes that the topology of the infrastructures and the initial failure rates of the components are available to the infrastructure managers, while the supply demand characteristics are unknown and simulated via multiple scenarios. 
    \end{inparaenum}
    
\end{itemize}

 To implement the proposed generalized vulnerability assessment framework, we leverage Bayesian network (BN) to model the dynamic fault tree (DFT) for each infrastructure system, which can efficiently capture the failure propagation dynamics within a particular infrastructure referred to as the \textit{intra-infrastructure failure}. Then, over a grid based geo-map, a network of interdependent infrastructure networks is created which can quantitatively measure the \textit{physical, logical and geographical} interdependencies that may exist between every pair of infrastructure systems. Thereafter, using a simulation based approach, the failure dynamics of the components of the child infrastructure that may be induced by the failure of interdependent infrastructure systems, referred to as the \textit{inter-infrastructure failure} is calculated. Finally, the compound risk of failure, i.e. the comprehensive vulnerability of a network of infrastructure systems, which may be caused by either the intra-infrastructure failure or inter-infrastructure failure is computed. The proposed framework is implemented and validated using existing synthetic data on electricity, water and supply chain networks.
The remainder of the paper is organized as follows. In Section~\ref{method}, the detailed methodology employed in this paper is described. In Section~\ref{data}, the data collection and pre-processing methods for the specific case study are described, followed by summarizing the results and key findings in the Section~\ref{result}. Finally the paper is concluded in Section~\ref{conc}.

\section{METHODOLOGY}
\label{method}
In this section, a detailed description of the methodology proposed in this paper to calculate the vulnerability of an infrastructure due to intra- and inter-infrastructure failures is presented. 
As mentioned before, the proposed framework is flexible enough to incorporate modifications in terms of the number and types of the infrastructure systems under consideration, any type of spatiotemporal characteristics, and even the lack of adequate data availability. First, to capture the failure propagation within an infrastructure, a DFT is constructed and then using dynamic Bayesian network, the time-dependent vulnerability is modeled (Section~\ref{intra-vulnerability}). To model the inter-infrastructure vulnerability where the \textit{exact supply-flow characteristics are not known}, a heuristic network of networks is created considering the physical, logical and geographical dependencies between the infrastructures (Section~\ref{inter-vulnerability}). Thereafter, leveraging a simulation approach, multiple cases are constructed to calculate the best, average or worst case inter-infrastructure failures. Finally, the intra-infrastructure and inter-infrastructure failures are combined to estimate the comprehensive vulnerability of infrastructure (Section~\ref{comp-vulnerability}).
The three parts of the overall methodology is depicted in Fig.~\ref{frame}. Details of each part is described in the following subsections.

\begin{figure}[ht]
    \centering
    \includegraphics[height = 2.1in, width=0.7\linewidth]{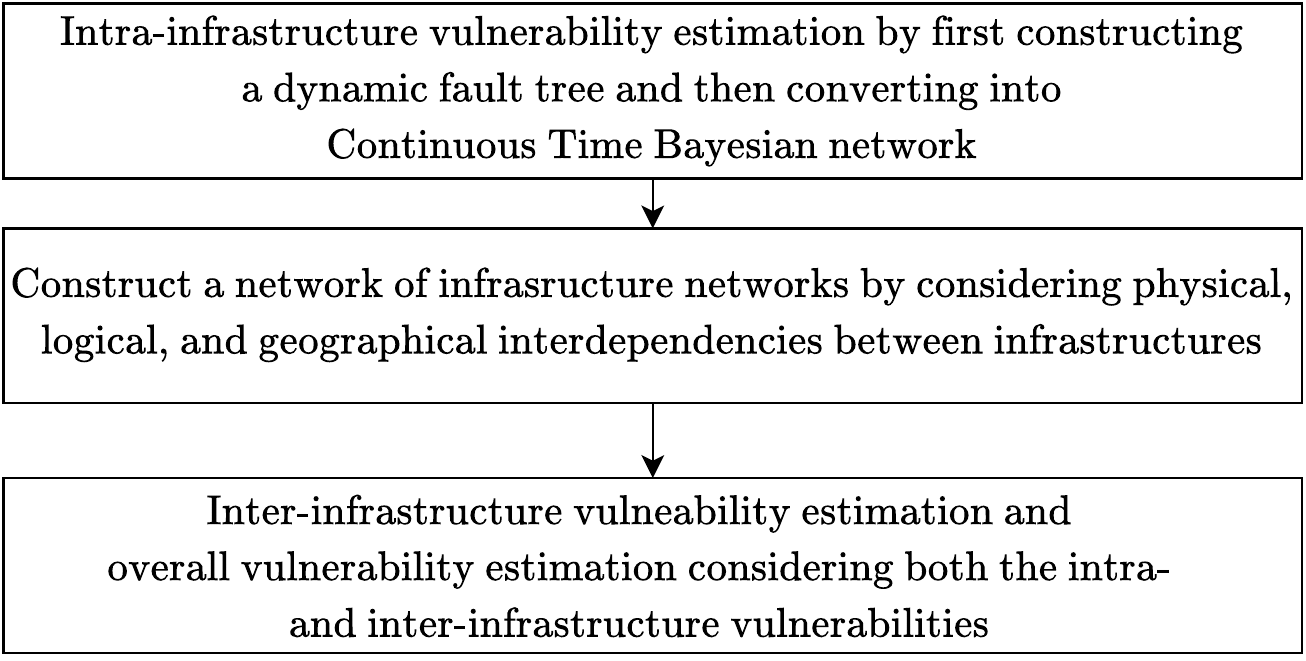}
    \caption{The overall research framework.}
    \label{frame}
\end{figure}

\subsection{Intra-infrastructure failure probability }
\label{intra-vulnerability}
In this section, the proposed method of leveraging dynamic fault tree (DFT) and continuous time Bayesian network to estimate inter-infrastructure vulnerability under incomplete information is described. A DFT consists of several gates as depicted in Fig.~\ref{dfta}. In a DFT, if two inputs $A$ and $B$ are connected via \textbf{AND} (Fig.~\ref{dfta} (i)) gate to produce output $X$, then $X$ fails if both $A$ and $B$ fail; on the other hand, if $A$ and $B$ are connected via \textbf{OR} gate (Fig.~\ref{dfta} (ii)) to produce $X$, then $X$ fails if either $A$ or $B$ fails. In a \textbf{PAND} or priority AND gate (Fig.~\ref{dfta} (iii)) , the output fails if all of its inputs fail and fail from left to right as depicted in the diagram. In a \textbf{$K/M$ VOTING }gate (Fig.~\ref{dfta} (iv) depicting a $2/3$ VOTING gate), the output fails, if at least $K$ of its $M$  inputs fail. In Fig.~\ref{dftb}, we showed how a $K/M$ VOTING gate can be modeled using a combination of basic AND and OR gates. In a $K/M$ VOTING gate with $M$ inputs, the output fails if at least $K$ of the $M$ inputs fail where $K \leq M$. It is noteworthy that the VOTING gate is a derived gate in a sense that a VOTING gate can be easily constructed using the basic AND and OR gate. If $M=2$, and $K=1$ then the VOTING gate is an OR gate and if $K=2$, then the VOTING gate is an AND gate. Essentially, a VOTING gate is an OR combination of the $\binom{M}{K}+ \binom{M}{K+1}+ ...+ \binom{M}{M}$ number of AND combinations of the inputs. Out of the total $M$ inputs, if any combination of $K$ such inputs fail, or any combination of $K+1$ to $M$ of such combinations fail, then the output fails. In a \textbf{WSP} or warm spare gate (Fig.~\ref{dfta} (v)), there is one primary input ($A$) and multiple spare inputs ($B$ and $C$). Initially, the primary input is switched on and the spares operate in a dormant or standby mode where, the failure rate of the spare is reduced by a factor $\alpha \in [0,1]$ called the dormancy factor ($\alpha$). When the primary unit fails, the first available spare becomes active. The output $X$ fails if all of the inputs fail. 

\begin{figure}[h]
    \begin{subfigure}[t]{0.35\textwidth}
     \raggedleft
        \includegraphics[height=3.25in]{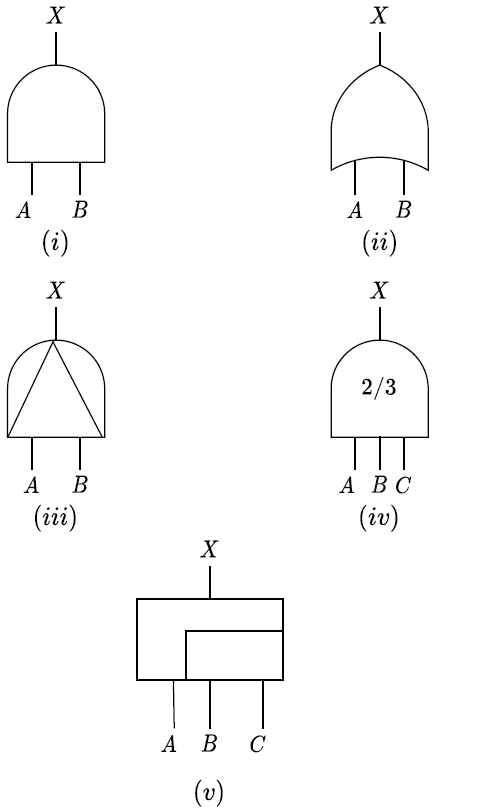}
        \captionsetup{justification=centering}
        \caption{}
        \label{dfta}
    \end{subfigure}
    \begin{subfigure}[t]{0.6\textwidth}
     \raggedright
        \includegraphics[height=3.25in]{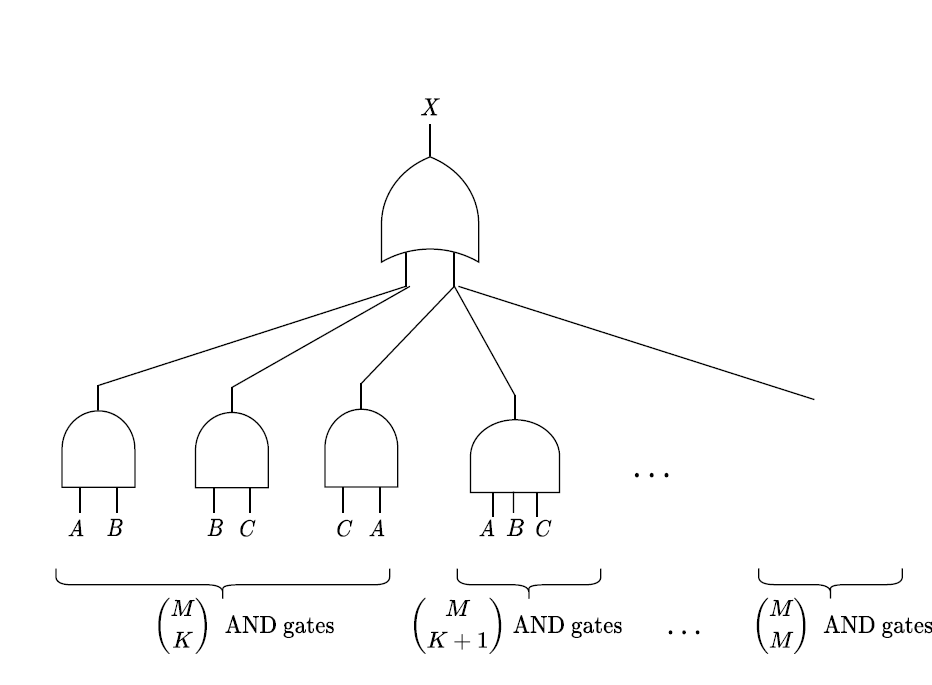}
        \captionsetup{justification=centering}
        \caption{}
        \label{dftb}
    \end{subfigure}
    \caption{(a) the gates in the DFT: (i) AND gate, (ii) OR gate, (iii) PAND gate, (iv) VOTING gate, and (v) WSP gate; (b) The depiction of VOTING gate using a combination of AND and OR gates}
\label{dft}
\end{figure}

\subsubsection{Constructing a dynamic fault tree }
\label{sdft1}
When any critical component of a particular infrastructure fails, then the failure propagates within the infrastructure network. In any networked infrastructure system, commodities or services flow from one component to another. Often, some type of commodities are generated at the source nodes and then it flows through multiple intermediate nodes before it is consumed or utilized by the customers in the terminal node. This general architecture is the backbone of the operation mode for several infrastructure systems like electricity distribution and transmission system, water distribution or the supply chain network. For example, in an electricity infrastructure, generators are connected to the loads via buses, and buses are connected via power lines in between them. The power generated at a generation station satisfies the loads connected to that particular bus and may be transferred to other buses via the transmission lines to satisfy the demands of the other buses. When a generator goes off, other generators connected to the bus try to make up for the potential demand of the bus~\cite{wood}. Thus, these generators can be considered as the sources of electricity flowing through the network, the buses can be considered as the intermediate nodes, and finally the loads connected to the buses can be considered as the consumers. Similarly, for water distribution system the water flows from the sources to the treatment plants, then to the storage reservoirs and finally to the distribution reservoirs from where the water is dispatched to the consumers~\cite{Sincero}. A similar network architecture prevails for a supply chain network, where the raw materials flow from the supplier to the manufacturer for processing the commodities and then the commodities are supplied to the retailers who supply the commodities to the customers~\cite{TANG2016}. 

\begin{figure*}[h]
    \begin{subfigure}[t]{0.49\textwidth}
     \raggedleft
        \includegraphics[height=2.5in]{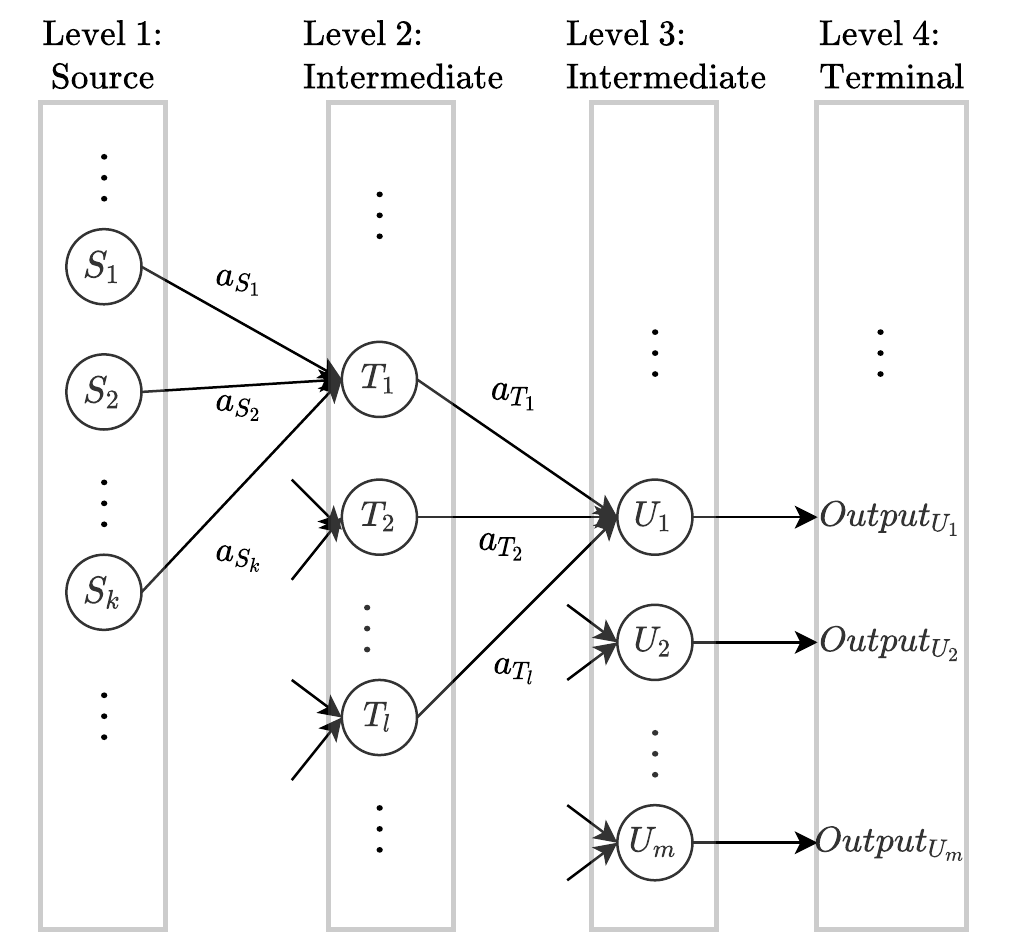}
        \captionsetup{justification=centering}
        \caption{}
        \label{a}
    \end{subfigure}
    \hspace{0.1 in}
    \begin{subfigure}[t]{0.49\textwidth}
     \raggedright
        \includegraphics[height=2.5in]{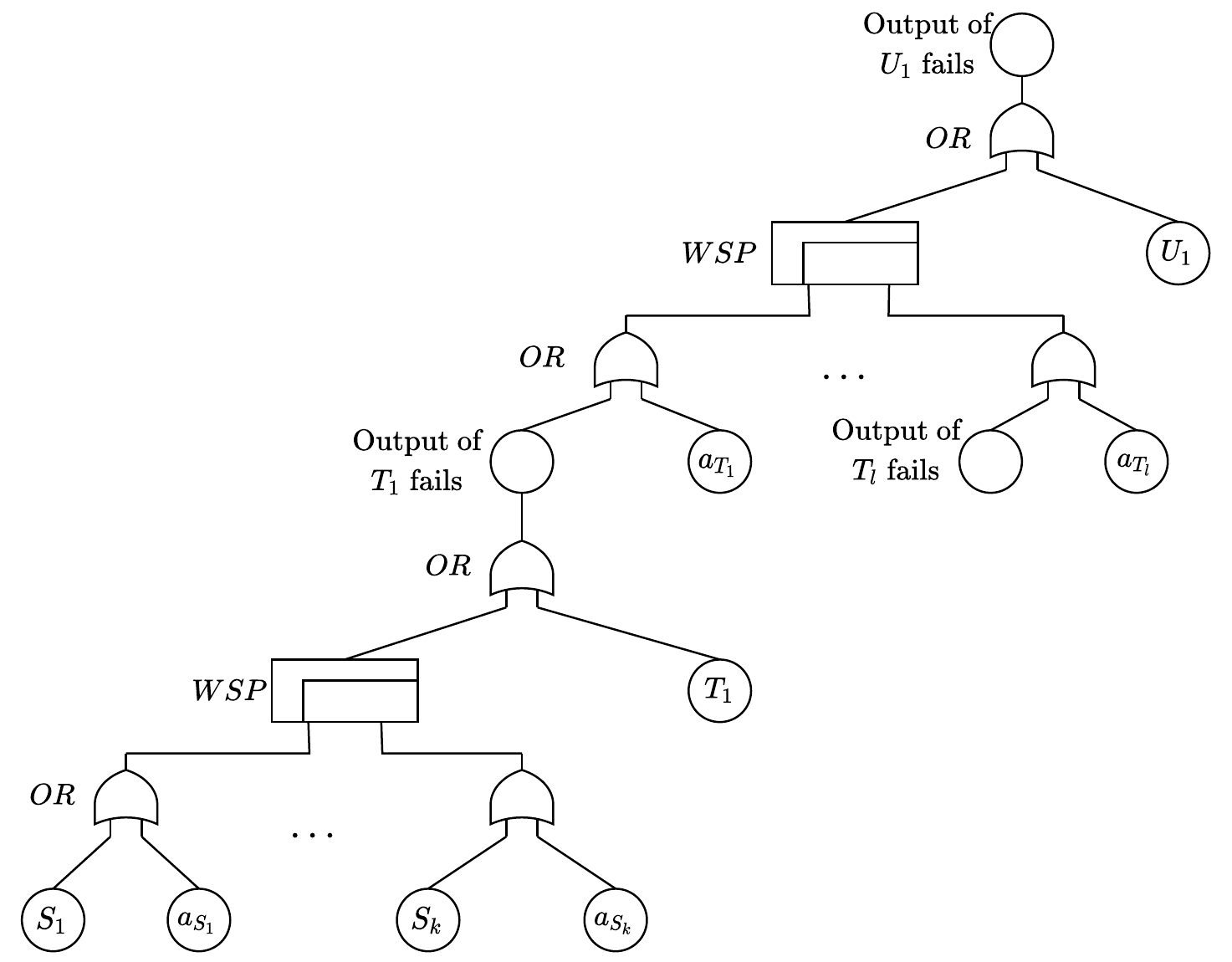}
        \captionsetup{justification=centering}
        \caption{}
        \label{b}
    \end{subfigure}
    \caption{(a) A multipartite directed acyclic graph constructed according to the material/service flow in a general infrastructure system, and (b) the dynamic fault tree (DFT) constructed from the network of the general infrastructure system under consideration.}
\label{in}
\end{figure*} 

Based on the flow of the commodities or services (e.g., electricity or water) from the source node to the terminal node via multiple intermediate nodes and connecting arcs, it is evident that \textit{different network infrastructure systems share a common layout}. In this paper, we present a \textit{generalized framework that can abstract a network for any such infrastructure when exact information on the nodes, arcs and interdependencies are not available}. First, according to the flow of the commodities or services within an infrastructure, a topological ordering~\cite{Cormen} of the nodes is created. From that topological ordering, a multipartite graph is constructed where the nodes of the graph are partitioned into multiple disjoint sets or levels say, $B_1, ..., B_l$. There exists directed edges from the nodes of level $B_b$ to $B_{b+1}$ for $b \in \{1,...,l-1\}$. For example, for a water distribution network, the first level consists of the sources, and the second level consists of the treatment plants. As water flows from the sources to the treatment plants, there exists directed edges from the nodes in the first level to the nodes in the second level. Similarly, edges exist between the nodes of the second level to the third level (e.g., reservoirs for water infrastructure) and finally to the layer of the terminal nodes (e.g., consumers). This layout can also be applied to several other infrastructure systems wherever there exists a flow of commodities or services from one set of nodes to the others like electricity infrastructure where electricity flows from generators to buses and loads. Such a common layout of a general infrastructure system has been depicted in Fig.~\ref{a}. 

Consider that the source nodes $S_1, S_2, ..., S_k$ are connected to the intermediate node $T_1$ via the arcs $a_{S_1}, a_{S_2}, ..., a_{S_k}$ respectively. If either $S_1$ or the arc coming out of $S_1$, i.e., $a_{S_1}$ fails, then the output of $S_1$ cannot reach $T_1$. Similarly, the output of $S_2$ cannot reach $T_1$ if either $S_2$ or the arc associated with it, $a_{S_2}$ fails; and if either $S_k$ or the arc associated with it, $a_{S_k}$ fails, then the output of the source $S_k$ cannot reach the destination $T_1$. When one of the sources in $S_1, S_2, ..., S_k$ or the associated arcs $a_{S_1}, a_{S_2}, ..., a_{S_k}$ fails, then the other sources try to compensate for the loss owing to the breakdown of such a source. Hence, the other sources start working in a more stressful condition than the dormant state to make up for the loss of supply. Due to this, the failure rate of the currently functional units which are in the standby mode increase by a factor ($\alpha$). This situation can be modeled leveraging a DFT, using the warm spare (WSP) gate. Using a DFT, the infrastructure network described in Fig.~\ref{a}, can be modeled using Fig.~\ref{b}. If any source of $S_1, S_2, ..., S_k$ or the associated arc with it $a_{S_1}, a_{S_2}, ..., a_{S_k}$ fails, then the output of the particular source cannot reach the intermediate node $T_1$. If all the sources connected with $T_1$ fail, then the input to $T_1$ fails. Furthermore, if one input fails, the failure rate of the other inputs increase by a certain factor. Hence, the input to $T_1$ failing is a result of a WSP gate output. Then, if the input to $T_1$ fails or the node $T_1$ itself fails, the output of $T_1$ fails. If either the output of $T_1$ fails or the arc coming out of $T_1$, i.e., $a_{T_1}$ fails, then the output of $T_1$ cannot reach the destination $U_1$. Similar to the previous layer, there are multiple nodes $T_1, T_2, ..., T_l$ connected to $U_1$ via the arcs $a_{T_1}, a_{T_2}, ..., a_{T_l}$. If any of the nodes of $T_1, T_2, ..., T_l$ or the arcs associated fails, then the output of that node may not reach $U_1$. Using a similar logic of the previous layer, the failure of the input to $U_1$ is modeled as a WSP gate output of its parent nodes or the incoming edges. Therefore, in this architecture, the output of $U_1$ fails if either the input of $U_1$ fails or the node $U_1$ itself fails. Thus, a recursive procedure where according to the commodity or service flow within an infrastructure network consisting of nodes and arcs can be converted into a dynamic fault tree (DFT) is proposed.

Though such an algorithm can be used to construct a DFT for any infrastructure network, such a DFT is not accurate. A more complex and accurate DFT can be constructed considering the accurate topology and the service-demand flow of the infrastructure, if they exist. However, the accurate DFT can be easily substituted in the overall framework of the paper as described in Fig.~\ref{frame}, without loss of generality when data is available. 

\subsubsection{Construction of a continuous-time Bayesian network}
\label{ctbn}
In order to construct a Continuous Time Bayesian Network (CTBN) from the DFT of the infrastructure~\cite{Boudali}, we extended the methodology proposed by Boudali \textit{et al.} in this paper. In the CTBN, the state space is continuous, which depict the failure time of a component of the DFT. In a CTBN, the random variable $\xi$ is in state $x$ means that the system component represented by $\xi$ failed in the time instant $x$, where $x$ is a non-negative real number. Here, the closed form solutions of the probabilities of failure of the output for each of the gates used in our DFT are presented. A detailed derivation is provided in the Appendix Section~\ref{app}. Considering a two input AND, OR and WSP gate, Table~\ref{dftgate} depicts the probabilities of failure of the output in time $[0,t]$. 
It is considered that the inputs to the gates are $A$ and $B$ where, the time of failure of $A$ follows exponential distribution with rate $\lambda_A$, and the time of failure of $B$ follows exponential distribution with rate $\lambda_B$. Then, for the WSP gate if $B$ is the spare unit with dormancy factor $\alpha$, the probability of failure of the output $X$ in time $[0,t]$ is depicted in the following Table~\ref{dftgate}. 
For the WSP gate, while deriving the closed form solution of the probability of failure, Boudali \textit{et al.} considered both the inputs $A$ and $B$ to have the same failure rate $\lambda$. This assumption is not practical in a real life scenario. In this study, the work is extended by considering the inputs $A$ and $B$ having different rates of failure: $\lambda_A$ and $\lambda_B$ respectively. The detailed derivation is provided in Appendix Section~\ref{app}. It is noteworthy that, when $\lambda_A = \lambda_B = \lambda$, our closed form solution reduces to the solution provided by Boudali \textit{et al.}~\cite{Boudali}. The dormancy factor ($\alpha$) for a two input WSP gate is considered as $0.5$. 

\begin{table}[h]
\centering
\begin{adjustbox}{max width=\linewidth,center}
 \begin{tabular}{|c | c |} 
 \hline
Gate & Probability of failure of output in $[0,t]$ \\
\hline
 AND & $1- e^{-\lambda_At} - e^{-\lambda_Bt} + e^{-(\lambda_A+\lambda_B)t}$\\
\hline
OR & $1-  e^{-(\lambda_A+\lambda_B)t}$\\
\hline
WSP & $(1-  e^{-\lambda_A t})(\frac{\alpha \lambda_B (e^{-t(\lambda_A + \alpha \lambda_B)}-1)}{-\lambda_A - \alpha \lambda_B} - \frac{\lambda_A(-(e^{-\lambda_B t}-1)(\lambda_A + \alpha \lambda_B) + \lambda_B e^{-t(\lambda_A + \alpha \lambda_B)}- \lambda_B)}{(\lambda_A + \alpha \lambda_B) (\lambda_B - \lambda_A -\alpha \lambda_B)})$\\
\hline
 \end{tabular}
 \end{adjustbox}
\caption{For each gate of our DFT, the probabilities of failure of the output in time $[0,t]$.}
\label{dftgate}
\end{table}

To summarize, the overall procedure to estimate the probabilities of intra-infrastructure failure 
has been depicted in  Algorithm~\ref{algo:intra}. 

\begin{algorithm}[!t]
\DontPrintSemicolon
\footnotesize
\textbf{Input}: An infrastructure $I_i$ with the nodes and the edges connecting the nodes. and the initial probabilities of failure for each node  and arc of $I_i$ denoted here as $\kappa$.\\
\textbf{Output}: For each node of $I_i$, the probabilities of failure arising from intra-infrastructure cascade. \\
\Begin{
/* Construction of the dynamic fault tree (DFT) for a general infrastructure */\;
\nl According to the material/service flow within $I_i$ find a topological ordering of the nodes of $I_i$\; 
\nl According to the topological ordering, partition the nodes of $I_i$ into disjoint sets (levels): $B_1, ..., B_l$, such that there exists arcs from nodes in $B_b$ to $B_{b+1}$ \; 
\nl \For{every node $n$ and arc $a$ of $I_i$}{
\nl Assign the non-zero initial probabilities of failure $\kappa(n)$ and $\kappa(a)$\;}
\nl \For{every level $lv \in \{2,..., b\}$}{
\nl \For{every node $T_l \in lv$}{
\nl Input to $T_l$ fails $\leftarrow$ $WSP$($\,OR$($\kappa(S_k), \kappa(a_{S_k})$)) , where $S_k$ are the parent nodes of $T_l$ in the level $lv-1$ and $a_{S_k}$ are the arcs between $S_k$ and $T_l$  \;
\nl Output of $T_l$ fails $\leftarrow$ $OR$(Input to $T_l$ fails, $\kappa(T_l$))\;
\nl $\kappa(T_l) \leftarrow$ Output of $T_l$ fails 
}}
/* Calculating probabilities of failure using continuous time Bayesian network */\;
\nl \For{every node in the DFT starting from the bottom level}{
\nl Using the closed form solutions in section~\ref{ctbn}, calculate the probabilities of failure of intermediate and top events}
}
\caption{{\texttt{IntraModel}$(I_i$, $\kappa$)}}
\label{algo:intra}
\vspace{-0.2cm}
\end{algorithm}

\subsection{Constructing interdependent network of networks}
\label{inter-vulnerability}
\begin{figure*}[ht]
    \centering
    \includegraphics[height=2.85in, width=0.85\textwidth]{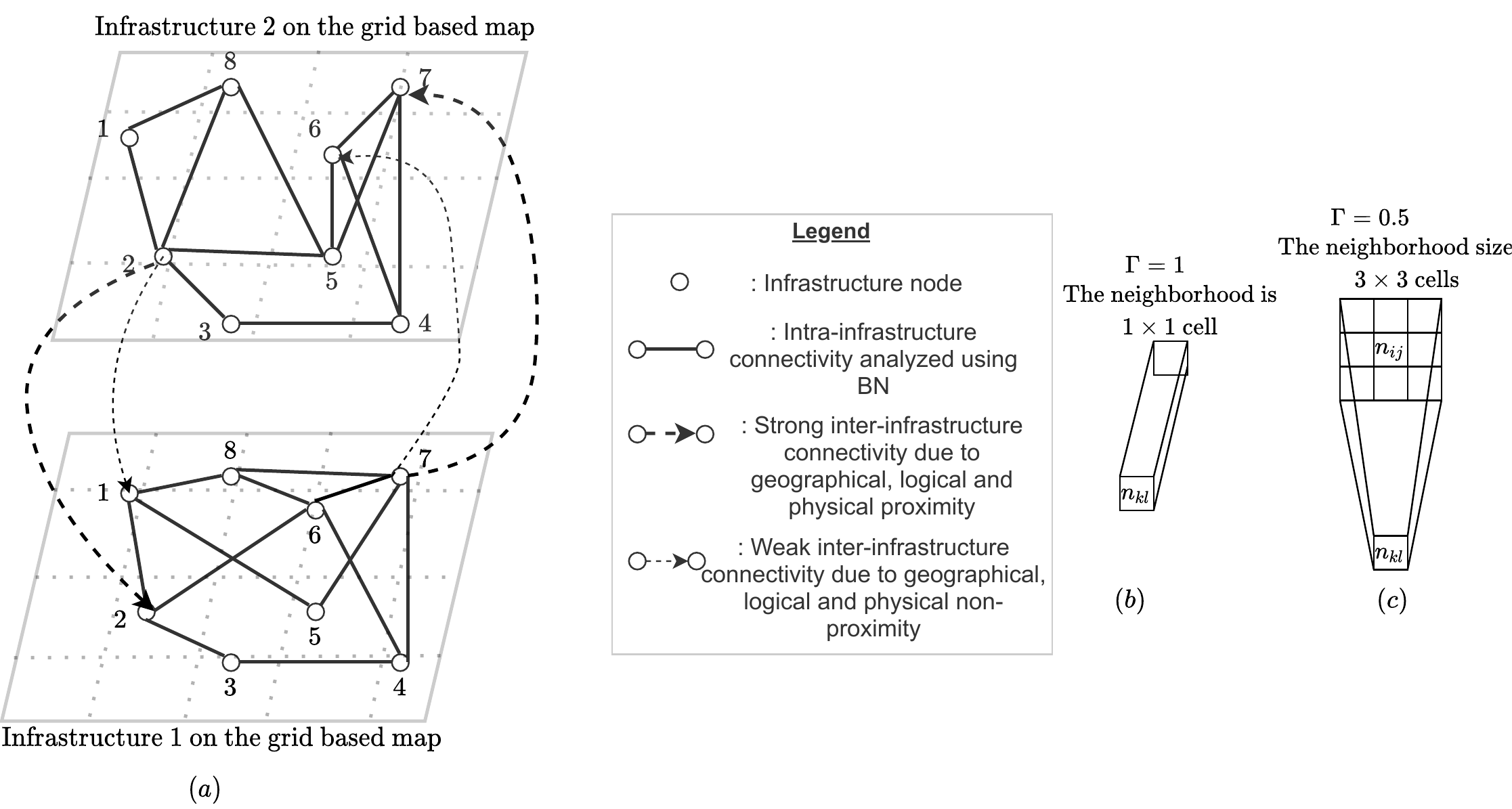}
    \caption{ The network of infrastructure networks on grid based map with intra-infrastructure and inter-infrastructure edges and the changes in neighborhood size as per $\Gamma$.}
    \label{Netofnet}
\end{figure*}

In the second step of the overall research framework where to estimate the \textit{inter-infrastructure vulnerability}, first a framework to construct a network-of-networks which can capture the physical, logical, and geographical interdependencies as depicted in the Algorithm~\ref{algo:cn} is proposed. To start with, it is considered that there are $s$ interdependent infrastructure systems, $I_1, I_2, ..., I_s$. Each infrastructure system is a network of components. Let us consider that there are $N_1$ number of nodes in $I_1$, $N_2$ number of nodes in $I_2$, ..., $N_s$ number of nodes in $I_s$. The $j^{th}$ node of infrastructure $I_i$ is denoted by $n_{ij}$. First, the algorithm checks if infrastructure $I_k$ is \textbf{logically} dependent on the infrastructure $I_i$. For example, the water distribution infrastructure may be physically and logically dependent on the electricity infrastructure system in a sense that the pumps of the water distribution infrastructure require electricity services to operate. If infrastructure $I_k$ is dependent on the infrastructure $I_i$, then there exists interdependency edges from components of $I_i$ to $I_k$. Furthermore, note that even if $I_k$ is dependent on $I_i$, all the nodes of $I_k$ may not require services from $I_i$; and all the nodes of $I_i$ may not provide service to a node of $I_k$. As an example, for electricity infrastructure to operate, water is required for cooling and also for the generators; however, all the nodes of a water distribution infrastructure cannot provide the required service to the electricity infrastructure. Water can be dispatched from the reservoirs but not directly from the treatment plants. These physical and logical dependencies are specific to infrastructure systems under consideration and expert opinion may be helpful in this regard. However, without loss of generality of the overall framework, these interdependency constraints can be relaxed or tightened according to the information available to the managers. In Algorithm~\ref{algo:cn}, it is considered if any node $n_{kl} \in I_k$ requires some type of service from a node $n_{ij} \in I_i$, then an interdependency edge can be added from $n_{ij}$ to $n_{kl}$. 

The infrastructure networks are placed on a grid based geo-map as shown in Fig. \ref{Netofnet} using QGIS~\cite{qgis}. The grid cells have a specific size, which is to be kept constant for all the infrastructures. The grid cells with the smallest resolution essentially represent the geographical coordinate of each node of the infrastructure. Using a larger resolution, it is allowed to discretize the geographical space into cells where each node of infrastructures belong to a specific cell of the grid. The cell based distance between every node ($n_{ij}$) of $I_i$ and every node ($n_{kl}$) of $I_k$ is calculated as the Euclidean distance. By add-one smoothing, the effective distance is modified such that the distance, $d \geq 1$. The interdependency edge strength between the components $n_{ij}$ and $n_{kl}$ is calculated as $1/d$. If the components $n_{ij}$ and $n_{kl}$ are physically or geographically distant, then the cell based distance ($d$) will be more, which in turn results in low interdependency edge strength. On the other hand, if the nodes $n_{ij}$ and $n_{kl}$ are in physical or geographical proximity, then the cell based distance ($d$) will be less, which results in a strong interdependent edge from $n_{ij}$ to $n_{kl}$. In this way, the \textbf{physical} and \textbf{geographical} interdependencies are taken into consideration. Finally, the strong interdependency edges are considered and the weak interdependency edges are removed from the system by using the threshold $\Gamma$, a user-defined parameter. The interdependency edges from $n_{ij}$ to $n_{kl}$ are only added if the interdependency edge strength ($\gamma_{kl}^{ij}$), calculated according to the physical and geographical proximity, is greater than or equal to the threshold $\Gamma$. However, it is noteworthy that for a large threshold $\Gamma$, there can be a situation where no interdependent edges can be added because all the interdependency edge strengths are less than $\Gamma$. Such a situation may not be practical for modeling the interdependent infrastructure systems. For example, consider the dependency of water distribution infrastructure on the electricity infrastructure system. All the water pumps require electricity to operate; however, due to high value of $\Gamma$, if no interdependency edge can be added from any bus of the electricity infrastructure to the water pumping stations, then it will not be practical to model the interdependency. Hence, in such a situation the interdependency edge of strength $\gamma_{kl}^{ij} = \frac{1}{ d}$ from the nearest node $n_{ij} \in I_i$ to the node $n_{kl}$ is added. Here, $d$ is the add-one smoothed Euclidean distance between the cell of $n_{ij}$ and $n_{kl}$, even though $\gamma_{kl}^{ij} < \Gamma$.

\begin{algorithm}[!t]
\DontPrintSemicolon
\footnotesize
\textbf{Input}: The infrastructure networks $I_1, I_2, ..., I_s$, the threshold $\Gamma \in [0,1]$ of the interdependency edge strength.\\
\textbf{Output}: A network of infrastructure networks ($G$) with the intra-infrastructure and inter-infrastructure links. \\
\Begin{
\nl $G = I_1 \cup I_2 \cup ... \cup I_s$\;
\nl        \For{every two infrastructure networks $I_i$ and $I_k$}{
\nl        \For{each node $n_{kl}$ of $I_k$}{
\nl        \For{each node $n_{ij}$ of $I_i$}{
\nl \If{the node $n_{kl}$ may require some service from the node $n_{ij}$ \textbf{or} the state of the node $n_{kl}$ may be affected by the state of the node $n_{ij}$}{
\nl        $d$ = (The Euclidean cell distance between  $n_{ij}$ and $n_{kl}$) + 1\;
\nl      Interdependency strength ($\gamma_{kl}^{ij}$) = $\frac{1}{d}$ \;
\nl     \If{$\gamma_{kl}^{ij} \geq \Gamma$}{
\nl      Add the interdependency edge of strength $\gamma_{kl}^{ij}$ from $n_{ij}$ to $n_{kl}$ in $G$}
         }\;
         }\;
\nl  \If{the node $n_{kl}$ may require some service from at least one node in $I_i$ \textbf{or} the state of the node $n_{kl}$ may be affected by the state of the node $n_{ij}$}{
\nl \If{$\nexists$ an edge from atleast one node of $I_i$ to $n_{kl}$}{
\nl Add the interdependency edge of highest strength = $\frac{1}{ d + 1}$ from $n_{ij}$ to $n_{kl}$ in $G$; where $n_{ij}$ is the nearest node of $I_i$ from $n_{kl}$, and $d = \text{The Euclidean cell distance between } n_{ij} \text{ and } n_{kl}$}}
         }
}

\nl \textbf{return} $G$
}
\caption{{\texttt{CreateNetwork} $(I_1, I_2, ..., I_s, \Gamma)$}}
\label{algo:cn}
\vspace{-0.2cm}
\end{algorithm}
\normalsize

\subsection{Infrastructure systems vulnerability assessment}
\label{comp-vulnerability}
After the network of infrastructure networks ($G$) is obtained using Algorithm~\ref{algo:cn}, a framework to estimate the vulnerability of each infrastructure due to both the intra- and inter-infrastructure cascading failures is created. 
In this process, three probabilities of failure are considered including, 
\begin{inparaenum}[(i)]
\item $P_{ij}^{intra}$--- the probability of failure of node $n_{ij}$ due to intra-infrastructure cascade effect at time $[0,t]$ where, $t \in [0,24]$; 
\item $P_{ij}^{inter}$--- the probability of failure of node $n_{ij}$ due to inter-infrastructure failure propagation from one infrastructure to another at time $[0,t]$ where, $t \in [0,24]$; and
\item $P_{ij}^{fail}$--- the comprehensive probability of failure arising from both intra-infrastructure and inter-infrastructure cascade at time $[0,t]$ where, $t \in [0,24]$. It is assumed that for a particular node $n_{ij}$, the intra-infrastructure failure probability is independent of the inter-infrastructure failure probability of the same node $n_{ij}$. This is a reasonable assumption because the source of the intra-infrastructure failure probability is any parent node of $n_{ij}$ within the same infrastructure $I_i$, whereas the inter-infrastructure failure probability is obtained from a node of a different network. Using the independence assumption, the probability of failure occurring from both the intra-infrastructure and inter-infrastructure cascades is calculated as, $P_{ij}^{intra}*P_{ij}^{inter}$.
\end{inparaenum}

Furthermore, when an infrastructure ($I_k$) receives service from another infrastructure ($I_i$), there can be generally three cases that may occur if multiple interdependency links exist. Say, the $\mathcal{N}_i$ be the set of nodes of $I_i$ from which there exists interdependency links to the node $n_{kl}$ of $I_k$ (i.e., $\exists n_{ij} \in \mathcal{N}_i$ such that $n_{ij} \rightarrow n_{kl}$ exists). In this scenario, there may be any of the following three cases that may arise:
\begin{enumerate}
    \item \textit{best case: }The node $n_{kl}$ fails if all the nodes of $\mathcal{N}_i$ fail. In this case, the node $n_{kl}$ fails after the node in $\mathcal{N}_i$, with the minimum probability of failure, fails. 
    \item \textit{worst case: }The node $n_{kl}$ fails if any one of the nodes in $\mathcal{N}_i$ fails. In this case, the node $n_{kl}$ fails after the node in $\mathcal{N}_i$, with the maximum probability of failure, fails.
    \item \textit{average case: }The functionality of node $n_{kl}$ is partially dependent on every node of $\mathcal{N}_i$. For example, if the nodes in $\mathcal{N}_i$ include certain types of supply facility locations; $n_{kl}$ are demand nodes and the demand of these nodes can be satisfied by multiple service nodes operating in conjunction. Without loss of generality, it can be assumed $n_{kl}$ is equally dependent on all the nodes of $\mathcal{N}_i$. 
\end{enumerate}
The infrastructure manager considers either of these cases according to their \textit{risk perception and prior experience}. For example, a risk-averse infrastructure manager who is interested to design a robust system shall consider the \textit{worst case scenario} of the supply demand characteristics from other infrastructures. On the other hand, a risk-seeking and opportunistic manager may consider the \textit{best case scenario}; whereas, the \textit{average case scenario} is well suited for a risk-neutral design of interdependent infrastructure systems. After considering a particular scenario of supply demand characteristic, a simulation environment is created where the probability of failure is calculated for each node of the infrastructures and how the failure probabilities evolve over the days is analyzed. As mentioned before, the failure probabilities depict the probability of a component failure at time $[0,t]$ where $t \in [0,24]$, depicting the hourly probabilities of failure for each component. The intra-infrastructure failure probability for the \textit{first day (i.e., iteration $1$)} is computed considering the initial failure probability follows exponential distribution, as indicated in previous studies~\cite{Boudali}. Then, the inter-infrastructure failure probability of node $n_{kl}$ due to $n_{ij}$ is calculated as the product of the strength of the interdependency links between them ($\gamma$) and the probability of failure of the parent node $n_{ij}$. 
\[ V_{kl}^{ij} = \gamma_{kl}^{ij} * P_{ij}^{fail} \]
According to the case of simulation (best, worst or average), the inter-infrastructure failure probability is calculated. As mentioned before, if the best case scenario is considered, the failure of the node happens if all the interconnected parent nodes fail. Hence, the probability of failure is the minimum of the failure probabilities of the parent nodes. On the other hand, a worst case situation may happen when a node fails if any one of the parent nodes fail. In this case, the failure probability of the node is same as the maximum probability of failure of the parent nodes. Finally, there can be an average case, where the failure of a node is dependent on all the parent nodes. In this case, an uniform distribution over all the parent nodes is assumed, and the failure probability is estimated as the average of the probabilities of failure of all the parent nodes. If there are two infrastructures $I_i$ and $I_k$, such that there exists interdependency edges from the nodes of $I_i$ to $I_k$, we have described how to calculate the probabilities of failure of the nodes of $I_k$ induced by the failure of the nodes of $I_i$ (refer to Section~\ref{inter-vulnerability}). However, as described before, the infrastructure $I_k$ may be dependent on multiple infrastructure systems, not a single infrastructure $I_i$. Now, the importance of one infrastructure on another infrastructure may be different. For example, for three interdependent infrastructure systems, $I_i, I_k,$ and $I_m$, where $I_m$ is dependent on both $I_i$ and $I_k$, the importance of dependency of $I_m$ on $I_i$ may be more than the dependency of $I_m$ on $I_k$ and vice-versa. To capture the combined effect of all the other infrastructure systems on a particular infrastructure, we introduce the relative importance matrix $R$ of dimension $s\times s$, where $s$ is the total number of infrastructures considered in the analysis. The importance of infrastructure $I_i$ on $I_k$ (or, the importance of dependency of $I_k$ on $I_i$) is depicted by the entry in the $i^{th}$ row and $k^{th}$ column i.e., $R_{ik}$. Furthermore, it is considered that the construction of $R$ follows the following properties.  
\begin{itemize}
    \item The diagonal entries of $R$ are $0$, i.e., $R_{ii} = 0 \; \forall i \in \{1,...,s\}$, indicating a particular infrastructure is not dependent on itself for inter-infrastructure failure;
    \item The column sum of every column in $R$ is $1$, i.e., $\sum_{i}R_{ik} = 1 \; \forall k \in \{1, ..., s\}$, indicating that the importance of interdependencies of the infrastructure $I_k$ over the other infrastructures must sum up to $1$. Note that, without loss of generality if the column sum is not $1$, it can be normalized to $1$. 
\end{itemize}
Now, the inter-infrastructure probability of failure of the node $n_{kl}$ of infrastructure $I_k$ due to all the parent infrastructures in time $[0,t]$ for either the best, worst or average case scenario ($P_{kl}^{inter}$) can be estimated as, 
\[
P_{kl}^{inter} = \sum_{i=1}^sR_{ik}P_{kl}^i 
\]
where $P_{kl}^i$ is the inter-infrastructure failure probability of the node $n_{kl}$ induced by $I_i$ in either the best, worst or average case scenario.  

Finally, the probability of failure of a node may result from either intra-infrastructure or inter-infrastructure cascade or both. 

\[
P_{ij}^{fail} = P_{ij}^{intra} + P_{ij}^{inter} - (P_{ij}^{intra}*P_{ij}^{inter})
\]

After the probabilities of failure $P_{ij}^{fail}$ for the \textit{first day (i.e., interation 1} is realized, they are considered as the initial failure probabilities of the components for the \textit{next day (i.e., iteration 2)} and again the intra-infrastructure, inter-infrastructure and combined/comprehensive probabilities of failures are calculated. This process is iterated over multiple iterations depicting the effects of multiple days of service outage of interdependent infrastructure systems. 
Finally, the sensitivity of the user-defined parameters considered in our model including, the threshold of the interdependency edge strength $\Gamma$, and the relative importance of one infrastructure on another $R_{ik}$ is considered. It is noteworthy that, as $\Gamma$ is changed, the inter-infrastructure probabilities of failure change corresponding to the best, worst or the average case scenarios under consideration. The following proposition~\ref{prop1} says that, with increasing $\Gamma$ (decrease in the neighborhood size), the inter-infrastructure failure probabilities will increase for the best case scenario, and the inter-infrastructure failure probabilities will decrease for the worst case scenario. In the best case scenario where the child node fails if all the parent nodes fail, the parent nodes act as redundancies. Hence, in that case, more the number of the parent nodes, less vulnerable is the child node. In the worst case scenario where the child node fails if any one of the parent nodes fail, the parent nodes act as non-redundancies or deficiencies. In such a scenario, more the number of parent nodes are, there is a higher chance of finding a critical node with high probability of failure which can cause the failure of the child node. Hence, in the worst case, more the number of the parent nodes are, more the vulnerable is the child node.  

\begin{proposition}
\label{prop1}
Consider two infrastructures $I_i$ and $I_k$, where there exists interdependency edges from $I_i$ to $I_k$. Say, $\mathcal{N}_{\Gamma z}$ be the neighborhood of a node $n_{kl} \in I_k$ such that the interdependency edge strength from any node $n_{ij} \in I_i$, is $\geq \Gamma z$ for a given $\Gamma z$. Furthermore, let us consider, $P_{\Gamma z}^b$, $P_{\Gamma z}^w$, and $P_{\Gamma z}^a$ denote the average inter-infrastructure failure probabilities of the nodes of $I_k$ for the given $\Gamma z$ in time $[0,t]$, for the best,worst and average case scenarios respectively. 
Then, 
\[ P_{\Gamma z1}^b \leq P_{\Gamma z2}^b, \; \text{where, } \Gamma z1 \leq \Gamma z2
\]

\[ P_{\Gamma z1}^w \geq P_{\Gamma z2}^w, \; \text{where, } \Gamma z1 \leq \Gamma z2
\]
\end{proposition}

\begin{proof}
To prove the proposition, the solutions for $\Gamma z2$ neighborhood is always contained in the solutions for $\Gamma z1$ neighborhood is shown. There can be four cases that may arise for $\Gamma z1 \leq \Gamma z2$ as follows. 
\begin{inparaenum}[(1)]
    \item If, $\mathcal{N}_{\Gamma z2} \neq \phi$, $\mathcal{N}_{\Gamma z1} \neq \phi$, in this case, according to the construction in Algorithm~\ref{algo:cn} if there exists some nodes $n_{ij} \in I_i$ in the $\Gamma z2$ neighborhood ($\mathcal{N}_{\Gamma z2}$) of $n_{kl}$, then that node must also exist in the $\Gamma z1$ neighborhood ($\mathcal{N}_{\Gamma z1}$) of $n_{kl}$. 
    \item If, $\mathcal{N}_{\Gamma z1} = \phi$ then, according to the construction, $\mathcal{N}_{\Gamma z2} = \phi$. In this case, let $\mathcal{N}_{\Gamma z1}^A$ and $\mathcal{N}_{\Gamma z2}^A$ be the auxiliary neighborhoods that are created by adding the then $n_{ih} \in I_i$ which is the nearest node of $n_{kl}$ to the neighborhoods $\mathcal{N}_{\Gamma z1}$ and $\mathcal{N}_{\Gamma z2}$. Hence both the $\mathcal{N}_{\Gamma z1}^A$ and $\mathcal{N}_{\Gamma z2}^A$ consists of $n_{ih}$.
    \item If, $\mathcal{N}_{\Gamma z2} = \phi$ and If, $\mathcal{N}_{\Gamma z1} \neq \phi$. In this case, say $n_{ih} \in I_i$ is the nearest node added to the $\mathcal{N}_{\Gamma z2} = \phi$ to ensure at least one interdependency edge exists, forming the auxiliary neighborhood $\mathcal{N}_{\Gamma z2}^A$. Now, as $\mathcal{N}_{\Gamma z1} \neq \phi$, then, we conclude $n_{ih}$ must belong to $\mathcal{N}_{\Gamma z1}$, i.e.,  $n_{ih} \in \mathcal{N}_{\Gamma z1}$. 
    \item If, $\mathcal{N}_{\Gamma z2} = \phi$ and if, $\mathcal{N}_{\Gamma z1} = \phi$, then similar to a previous case, let $n_{ih} \in I_i$ which is the nearest node of $n_{kl}$ is added to the neighborhoods $\mathcal{N}_{\Gamma z1}$ and $\mathcal{N}_{\Gamma z2}$ to form the auxiliary $\mathcal{N}_{\Gamma z1}^A$ and $\mathcal{N}_{\Gamma z2}^A$. Hence, both the $\mathcal{N}_{\Gamma z1}^A$ and $\mathcal{N}_{\Gamma z2}^A$ consists of $n_{ih}$. 
\end{inparaenum}

From the above four cases, we conclude that, either $\mathcal{N}_{\Gamma z2} \subseteq \mathcal{N}_{\Gamma z1}$ or, $\mathcal{N}_{\Gamma z2}^A \subseteq \mathcal{N}_{\Gamma z1}$ or, $\mathcal{N}_{\Gamma z2}^A \subseteq \mathcal{N}_{\Gamma z1}^A$. Now, according to the definition, in the best case the node $n_{kl}$ fails if all the nodes in the neighborhood fail. That is from Algorithm~\ref{algo:cv}, $P_{\Gamma z}^b = \argmin_{n_{ij} \in \mathcal{N}_{\Gamma z}}V_{ij}^{kl}$. As we have a minimization problem here, we conclude $P_{\Gamma z1}^b \leq P_{\Gamma z2}^b, \; \text{where, } \Gamma z1 \leq \Gamma z2$. Similarly, according to the definition, in the worst case the node $n_{kl}$ fails if at least one of the nodes in the neighborhood fails. That is from Algorithm~\ref{algo:cv}, $P_{\Gamma z}^b = \argmaxA_{n_{ij} \in \mathcal{N}_{\Gamma z}}V_{ij}^{kl}$. As we have a maximization problem in this scenario, we conclude $P_{\Gamma z1}^w \geq P_{\Gamma z2}^w, \; \text{where, } \Gamma z1 \leq \Gamma z2$. Hence, proved. 

\end{proof}

\footnotesize
\begin{algorithm}[H]
\DontPrintSemicolon
\footnotesize
\textbf{Input}: $G$, maximum number of iterations ($M \geq 1$), mean time of failure for the root nodes to start with ($\lambda$), and the relative importance matrix $R$  \\
\textbf{Output}:  The probability of failure $P_{ij}^{fail}$ for each node of each infrastructure in $G$\\
\Begin{

 \nl $iterations = 0$ \;
\nl \While{iterations $\leq M$}{
 
\nl \If{iterations = 0}{
\nl The initial probabilities of failure $\kappa$ $\sim$ Exp($\lambda$)\;}
\nl \Else{
\nl The initial probabilities of failure $\kappa$ $\leftarrow$ $P^{fail}$ of last iteration\;}
\nl \For{each infrastructure network $I_i  \in G$}{
 \nl  \For{each node $n_{ij} \in I_i$}{
 \nl   $P_{ij}^{intra} \leftarrow$ \texttt{IntraModel}($I_i$, $\kappa$)\;
 \nl   $P_{ij}^{fail} \leftarrow P_{ij}^{intra}$
}
}

\nl \For{every two network $I_i$ and $I_k$}{
\nl \If{$n_{ij} \rightarrow n_{kl}$ exists}{
\nl $V_{kl}^{ij} = \gamma_{kl}^{ij} * P_{ij}^{fail}$ \;
\nl \For{each node $n_{kl}$}{
/* Simulate either the best, average or worst case scenario */\;
\nl \If{Best case}{
\nl $P_{kl}^{i} = \argmin_{(j)}V_{kl}^{ij} $\;}
\nl \ElseIf{Worst case}{
\nl $P_{kl}^{i} = \argmaxA_{(j)}V_{kl}^{ij}$ \;}
\nl \ElseIf{Average case}{
\nl    /* Considering uniform distribution*/\;
\nl $P_{kl}^{i} =\frac{1}{|\mathcal{N}_k|} \sum_{n_{i,j}\in \mathcal{N}_k}V_{kl}^{ij}$\;}
}
}
}
\nl \For{every infrastructure $I_k$}{
\nl \For{every infrastructure $I_i$ from where interdependency edge to $I_k$ exists}{
\nl \For{every node $n_kl \in I_k$}{
\nl $P_{kl}^{inter} = \sum_{i = 1}^s R_{ik}P_{kl}^{i}$}}}
\nl \For{every infrastructure $I_k$}{
\nl \For{every node $n_{kl} \in I_k$}{
\nl $P_{kl}^{fail} =  P_{ij}^{intra} + P_{ij}^{inter} - (P_{ij}^{intra}*P_{ij}^{inter})$}}

\nl $iterations$ $=$ $iterations+1$

} 
\nl \For{each node $n_{ij}$ in each infrastructure network $I_i \in G$ }{
\nl \textbf{return} $P_{ij}^{final}$ \;}}

\caption{{\texttt{CalculateVulnerability} ($G, M, \lambda$, $R$) }}
\label{algo:cv}
\vspace{-0.2cm}
\end{algorithm}
\normalsize

\section{DATA COLLECTION AND PRE-PROCESSING}
\label{data}
To implement our proposed framework, a case study with three infrastructure systems: electricity, water, and a small scale supply chain network is considered. Various details on the data collection and pre-processing approach is presented in this section.

\subsection{Electricity infrastructure: } A typical electricity infrastructure system consists of generators, buses, lines and loads~\cite{wood}. The buses are connected to each other via lines and both the generators and the loads are connected to the buses. Each line is originating from a source bus and end up at a destination bus. The hypothetical node-network data representing a typical power grid is leveraged for our analysis. This data is obtained from the Reliability Test System Grid Modernization Lab Consortium (RTS GMLC) ~\cite{rtsgmlc}. The data is geo-coded and a failure probability is associated with each component (generators, buses, and lines). Using QGIS, a grid with $0.25$ degrees horizontal and vertical spacing is constructed. This grid over the maps of California, Arizona and Nevada as shown in Fig.~\ref{geodat1} is placed. The initial hourly failure rates of the components (generators, buses, and the lines) are obtained from the RTS-GMLC data set. A summary of the initial failure rates of the components is provided in Table~\ref{table:elec1}. 

\begin{figure}[h]
    \centering
    \includegraphics[ width=0.8\linewidth]{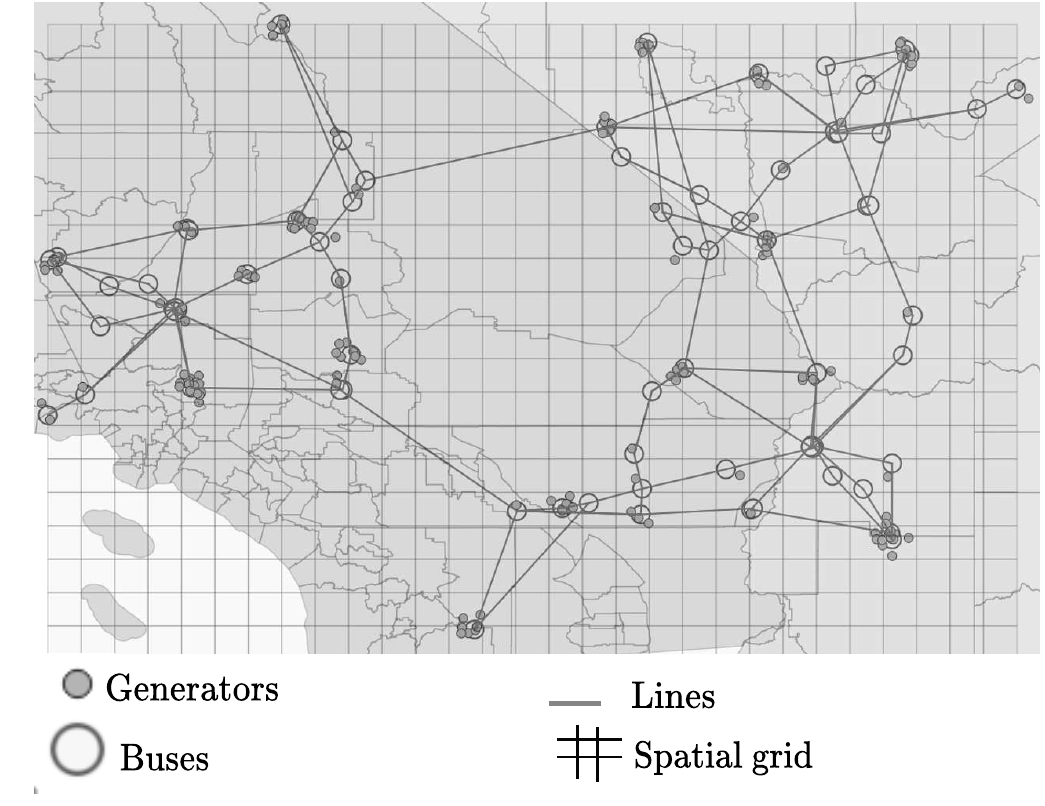}
    \caption{The layout of the electricity infrastructure system is depicted over the maps of the states of Arizona, California, and Nevada.}
    \label{geodat1}
\end{figure}

\begin{table}[h]
\centering
\begin{adjustbox}{max width=\linewidth,center}
 \begin{tabular}{||c | c | c | c | c ||} 
 \hline
Component & \makecell[c]{Mean \\failure \\ rate ($\lambda_{base}$)} & \makecell[c]{Maximum\\ failure \\ rate}  & \makecell[c]{Minimum\\ failure \\ rate} & \makecell[c]{Standard\\ deviation\\ ($\sigma_{base}$)} \\
\hline \hline
Generators & 0.003 & 0.012 & 0.001 & 0.001 \\
\hline
Lines & 0.007 & 0.02 & 0.001 & 0.003 \\ 
\hline
Buses & 0.002 & 0.01 & 0.001 & 0.001 \\
\hline
\end{tabular}
\end{adjustbox}
\caption{Summary of the initial failure rates of the components of the electricity infrastructure.}
\label{table:elec1}
\end{table}  

\subsection{Water distribution network:}
 In literature, for water distribution infrastructure, hypothetical networks like Anytown, Colorado Springs Utilities, EXNET or Richmond are extensively used with different types of topologies like grid iron, ring system, radial system, or dead end system~\cite{Mazumder, Yazdani}. However, none of these hypothetical networks are geo-coded. In a general water distribution network (WDN), water is brought to the treatment plants from the sources using pumping or gravity system. Then from the treatment plants, water is stored in storage reservoirs. Finally, the water is brought and temporarily stored in distribution reservoirs to meet the fluctuating demands~\cite{reservoirs1, Sincero}. It is considered that the water flows through the pipelines from one node to another via pumping system. Hence, at every stage of the distribution network, electricity supply is required for operation of a component of the WDN. In this study, a hypothetical WDN with a total of $30$ nodes, out of which three are sources, five are treatment plants, six storage reservoirs and sixteen distribution reservoirs with a topology similar to radial system is considered. The simulated WDN is placed on the same geographical boundary of the electricity infrastructure for ease of analysis. In Fig.~\ref{geodat2}, the spatial distribution of the WDN is shown.  In Table~\ref{table:water1}, the mean and the standard deviation of the hourly rates of failure for different components of the WDN are depicted. Note that, these hypothetical values are arbitrarily selected by the researchers of this study, and may not depict the actual failure rates of the components in a real life scenario. However, not to mention that these parameters can be updated easily in presence of actual data or expert opinion without loss of generality of the overall framework. 
\begin{figure}[H]
    \centering
    \includegraphics[ width=0.8\linewidth]{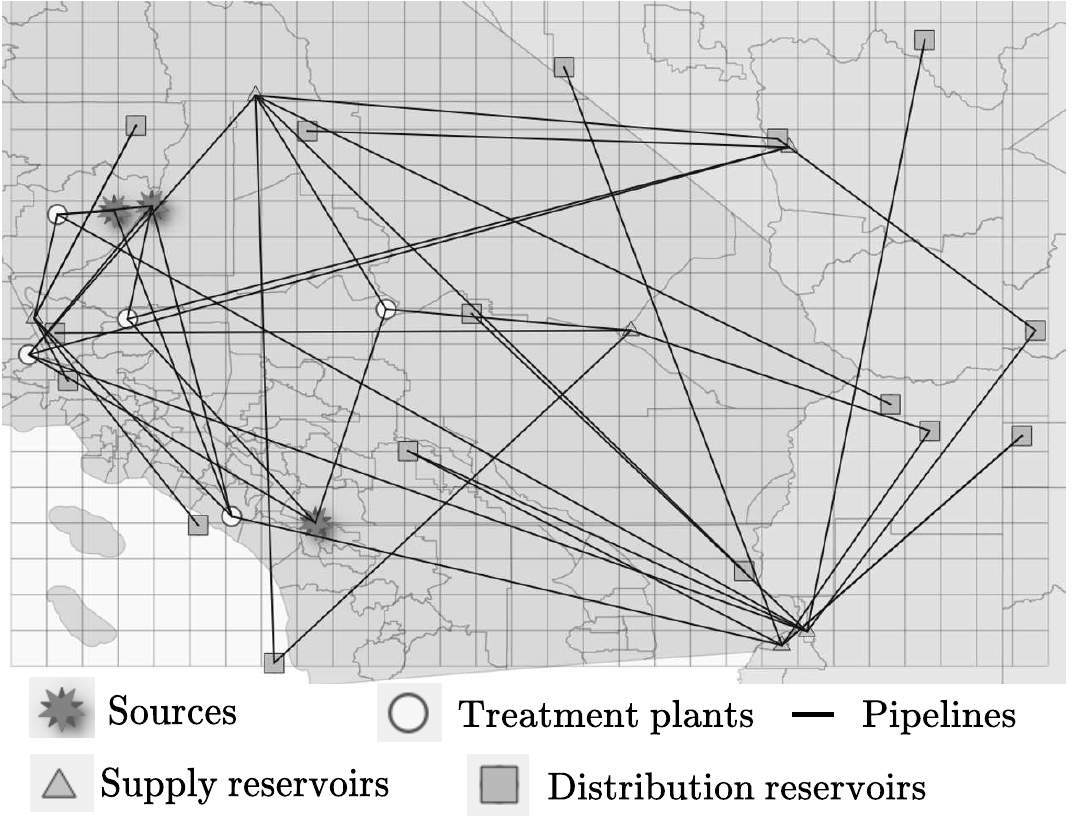}
    \caption{The layout of the water distribution system.}
    \label{geodat2}
\end{figure}

\begin{table}[H]
\centering
\begin{adjustbox}{max width=\textwidth,center}
 \begin{tabular}{||c | c | c | c | c ||} 
 \hline
Component & \makecell[c]{Mean failure \\ rate $(\lambda_{base})$} &  \makecell[c]{Standard \\deviation \\($\sigma_{base}$) } \\
\hline \hline
Sources & 0.005 & 0.001 \\
\hline
Treatment plants & 0.008 & 0.001 \\ 
\hline
Storage reservoirs & 0.009 & 0.002  \\
\hline
Distribution reservoirs & 0.01 & 0.002  \\
\hline
Pipelines & 0.01 & 0.001 \\
\hline
\end{tabular}
\end{adjustbox}
\caption{Summary of the initial failure rates of the components of the water distribution infrastructure.}
\label{table:water1}
\end{table}  

\subsection{Supply chain network: }
A typical supply chain network consists of suppliers, manufacturers and retailers~\cite{Mari}. Often other sets of nodes like plants or distributors are also considered as components of the supply chain network (SCN)~\cite{Gong}. However, without loss of generality, it is considered that a SCN consists of the suppliers, manufacturers and retailers. The material flow is between the suppliers and manufacturers, and between manufacturers and retailers as represented by the arcs of the network~\cite{TANG2016}. In this study, a hypothetical single commodity SCN with a total of fifteen nodes out of which, there are three supplier nodes, five manufacturer nodes and seven retailer nodes is considered.  The SCN over the same geo-grid of the electricity and water distribution network as depicted in Fig.~\ref{geodat3} is placed. In Table~\ref{table:supp1}, the mean and standard deviation of the initial hourly failure rates of the components of the hypothetical supply chain network considered in this study are summarized. Again, in presence of an actual network and data, the rates can be updated according to the user inputs. 

\begin{figure}[h]
    \centering
    \includegraphics[ width=0.8\linewidth]{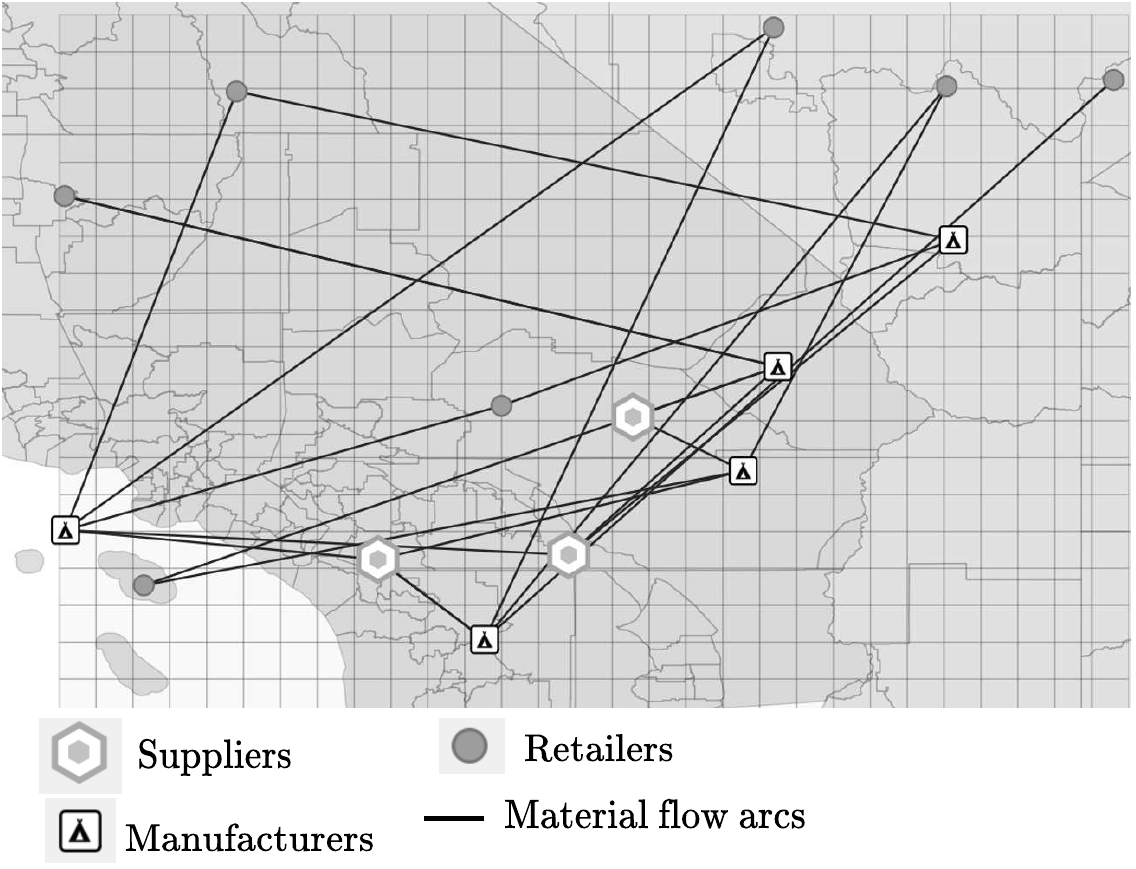}
    \caption{The layout of the supply chain network.}
    \label{geodat3}
\end{figure}

\begin{table}[h]
\centering
\begin{adjustbox}{max width=\linewidth,center}
 \begin{tabular}{||c | c | c | c | c ||} 
 \hline
Component & \makecell[c]{Mean failure \\ rate ($\lambda_{base}$)} &  \makecell[c]{Standard \\ deviation \\ ($\sigma_{base}$)} \\
\hline \hline
Suppliers & 0.005 & 0.001 \\
\hline
Manufacturers & 0.008 & 0.001 \\ 
\hline
Retailers & 0.009 & 0.002  \\
\hline
\makecell[c]{Material flow \\ arcs} & 0.01 & 0.001 \\
\hline
\end{tabular}
\end{adjustbox}
\caption{Summary of the initial failure rates of the components of the supply chain network.}
\label{table:supp1}
\end{table}  

\section{RESULTS}
\label{result}
Following the data pre-processing, the intra-infrastructure failure probabilities for all the three infrastructures using the DBN of the infrastructure as described in Algorithm~\ref{algo:intra} is estimated. Thereafter, the inter-infrastructure vulnerability and finally the comprehensive vulnerability resulting from both the intra-infrastructure and inter-infrastructure failures are estimated. As mentioned before, in this study the vulnerability of an infrastructure represents the dynamic probability of failure as a function of time, which is the failure probability of each of the infrastructure components in time $[0,t]$, where, $t \in [0,24]$. 

\subsection{Intra-infrastructure failure probability estimation}
In this section, our key findings on the failure probabilities of each infrastructure arising from the intra-infrastructure connections and the cascade propagation within a network is presented. For the \textit{electricity infrastructure ($I_1$)}, the mean failure rates of the different buses, generators and transmission lines are obtained from sampling with replacement from a normal distribution with mean $\lambda_{base}$ and standard deviation $\sigma_{base}$ as depicted in Table~\ref{table:elec1}. The Chi-squared goodness of fit test~\cite{Chi} is performed to identify that the mean rates of failure depicted in the RTS-GMLC data can be approximated using a normal distribution. For example, for the generators, the rate is assumed to be normally distributed with mean $0.003$ and standard deviation $0.001$, and the actual failure rates of each generator is obtained from sampling of this normal distribution. Similarly for the water distribution network ($I_2$) and the supply chain network ($I_3$), the initial failure rates of each individual component is obtained from sampling of a normal distribution with mean and standard deviation as depicted in the Tables~\ref{table:water1} and~\ref{table:supp1} respectively. The failure time of each infrastructure component follows exponential distribution with the realized rate of the particular component. That is, the initial probability of failure due to intra-infrastructure cascade in time $[0,t]$ is $P^{intra}(t) = 1- e^{-\lambda t}$, where $\lambda \sim \mathbf{N}(\lambda_{base}, \sigma_{base}^2)$. The distribution of $\lambda$ can be identified using some goodness-of-fit test. Furthermore, the sensitivity of the parameter $\lambda$ can be identified using some sensitivity analysis framework by alternating the mean of the distribution of $\lambda$ (for e.g., in case of a disaster, it is reasonable to assume that the mean failure rate increases, i.e., say, $\lambda \sim \mathbf{N}(\lambda_{base}+\sigma_{base}, \sigma_{base}^2)$. Such a framework is used in different areas like mental health prediction and crime analysis~\cite{mukherjee2021multi, ganguly2021multifaceted}. However, in this paper, the failure rate $\lambda$ is considered to follow a normal distribution with mean $\lambda_{base}$ and standard deviation $\sigma_{base}$ identified using Chi-squared goodness of fit test. Using Algorithm \ref{algo:intra} a DFT is constructed for each of the infrastructure and obtain the failure probabilities for each component of the networks due to the failure propagation within the network. The failure probabilities of each infrastructure components in time $[0,24]$ or within $24$ hours due to the cascading failure within the network has been depicted in Fig.~\ref{res_intra}. Furthermore, a diagnostic test is performed where it is found that the intra-infrastructure vulnerability of a node is particularly sensitive to the amount of redundancy associated with it, and the vulnerability decreases with increasing redundancy. The redundancy depicts the number of parent nodes associated with a particular node. This supports the findings from previous studies and thus validates our results~\cite{RE_book1, RE_book2}.
Intuitively, as there exists more parent nodes for a particular node, the failure probability of the input decreases, resulting in a decrease in the vulnerability of the child node as expected (see Fig~\ref{redund}).

\begin{figure*}[ht]
\centering
    \includegraphics[ width=\textwidth]{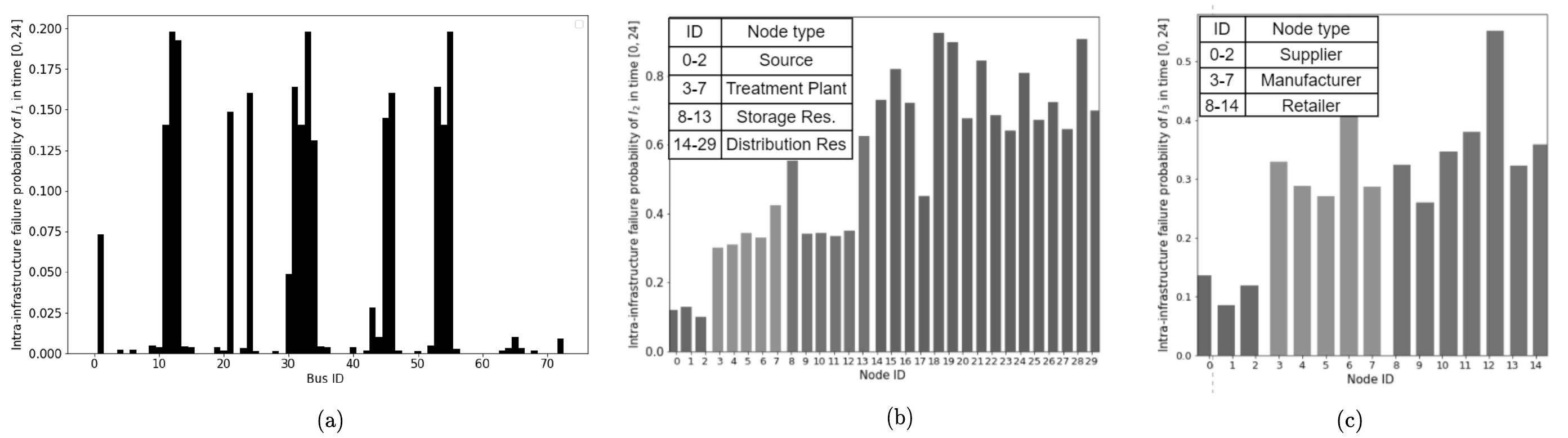}
    \caption{The component (node) wise probabilities of failure in time $[0,24]$ due to intra-infrastructure cascade for (a) the electricity infrastructure considering the buses as the components, (b) the water distribution network where the Node IDs: $0$ to $2$ represent the source nodes, $3$ to $7$ represent the treatment plants, $8$ to $13$ are the storage reservoirs, and $14$ to $29$ represent the distribution reservoirs and (c) the supply chain network, where the Node IDs: $0$ to $2$ represent the supplier nodes, $3$ to $7$ represent the manufacturer nodes, and $8$ to $14$ are the retailer nodes.}
    \label{res_intra}
\end{figure*}

\begin{figure*}[h]
    \centering
    \begin{subfigure}[t]{0.31\textwidth}
        \centering
        \includegraphics[width=\textwidth]{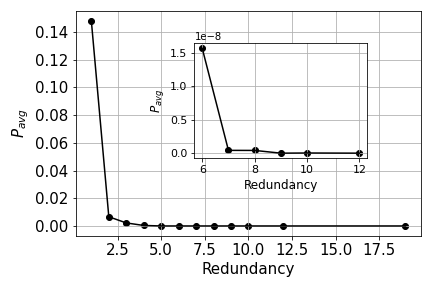}
        \captionsetup{justification=centering}
        \caption{}
        \label{red_el}
    \end{subfigure}
    \hfill
    \begin{subfigure}[t]{0.31\textwidth}
        \centering
        \includegraphics[width=\textwidth]{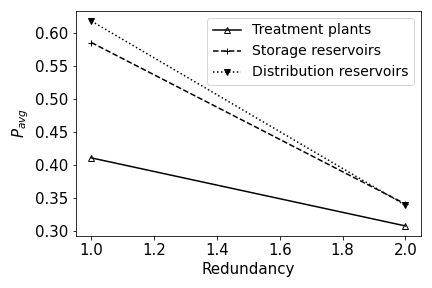}
        \captionsetup{justification=centering}
        \caption{}
        \label{red_wat}
    \end{subfigure}
    \hfill
    \begin{subfigure}[t]{0.31\textwidth}
        \centering
        \includegraphics[width=\textwidth]{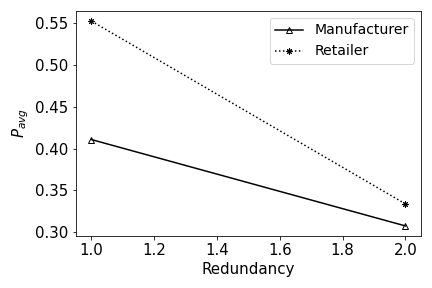}
        \captionsetup{justification=centering}
        \caption{}
        \label{red_scn}
    \end{subfigure}
    \caption{Redundancy versus intra-infrastructure failure probabilities of the components for (a) electricity infrastructure, (b) water distribution network and (c) supply chain network.}
\label{redund}
\end{figure*} 

\subsection{Inter-infrastructure vulnerability modeling}
 For an electricity infrastructure to operate successfully, it is assumed that water is required for power generating operations and cooling the system. Hence, every bus of the electricity infrastructure require certain type of service from the distribution reservoirs of the water distribution infrastructure. Similarly, for a water distribution network to operate, pumps require electricity~\cite{Rinaldi}. It is considered, that the pumps are associated with all the nodes of the water distribution infrastructure, i.e., the sources, treatment plants, storage reservoirs and distribution reservoirs. Hence all the nodes of the water distribution infrastructure require services from the electricity infrastructure. As a first step of estimating the inter-infrastructure vulnerability, the network of networks as depicted in Algorithm~\ref{algo:cn} is constructed. As the value of $\Gamma$ is changed, the number of interdependent edges from one infrastructure to the other changes. It is  identified that for high value of $\Gamma$, the neighborhood size of a cell is small and hence, the number of interdependent edges are also small; on the other hand, as $\Gamma$ decreases, the neighborhood size increases and there can be many potential nodes of a parent infrastructure which may provide service to a node of the child infrastructure. In Figs.~\ref{gam_a} and~\ref{gam_b} respectively, how the number of interdependent edges from the electricity infrastructure ($I_1$) to the water distribution infrastructure ($I_2$) and the supply chain network ($I_3$), and, from the water distribution network ($I_2$) to the electricity infrastructure ($I_1$) and the supply chain network ($I_3$) vary according to $\Gamma$ have been depicted.   

\begin{figure}[h]
    \centering
    \begin{subfigure}[t]{0.49\linewidth}
        \centering
        \includegraphics[width=\linewidth]{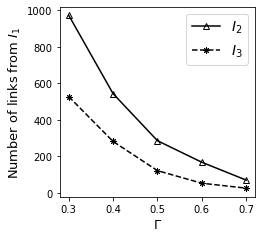}
        \captionsetup{justification=centering}
        \caption{}
        \label{gam_a}
    \end{subfigure}
    \hfill
    \begin{subfigure}[t]{0.49\linewidth}
        \centering
        \includegraphics[width=\linewidth]{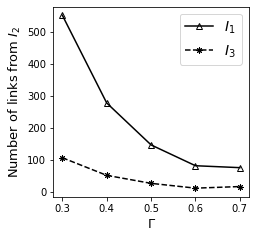}
        \captionsetup{justification=centering}
        \caption{}
        \label{gam_b}
    \end{subfigure}
    \caption{Variation in the number of interdependent edges versus $\Gamma$, for the parent infrastructure: (a) electricity distribution infrastructure, and (b) water distribution network.}
\label{in}
\end{figure} 

After constructing the network of networks for $\Gamma = 0.5$ (starting point), the inter-infrastructure probabilities of failure induced by one infrastructure to other as described in Algorithm~\ref{algo:cn} and Algorithm~\ref{algo:cv} is estimated.
Table~\ref{inter2} depicts the inter-infrastructure failure probabilities averaged over the nodes of the infrastructure induced by one infrastructure on the other corresponding to the best, worst and average case scenarios. The column $I_i \rightarrow I_k$ denotes the vulnerability (average probability of failure of the nodes in time $[0,24]$) of the infrastructure $I_k$ induced by the vulnerability of the interdependent infrastructure network ($I_i$), i.e., $\frac{\sum_{l=1}^{N_k} P_{kl}^{i}}{N_k}$ where, $N_k$ is the number of nodes in infrastructure $I_k$, and $P_{kl}^{i}$ is the inter-infrastructure probability of $l^{th}$ node $n_{kl} \in I_k$ induced by $I_i$. In this study, $I_1$ is the electricity infrastructure, $I_2$ is the water distribution network and $I_3$ is the supply chain network. For example, the column $I_1 \rightarrow I_2$ denotes the vulnerability (average probability of failure of the nodes in time $[0,24]$) of the water distribution infrastructure ($I_2$) induced by the vulnerability of electricity infrastructure network ($I_1$).  Various important phenomena regarding the inter-infrastructure failure probabilities for the best, worst and average case scenarios for different values of $\Gamma$ are observed. According to the Proposition~\ref{prop1}, as $\Gamma$ increases, the average inter-infrastructure failure probabilities of an infrastructure in time $[0,t]$ for the best case scenario increases. On the other hand, as $\Gamma$ increases, the inter-infrastructure failure probabilities for the worst case scenario in time $[0,t]$ decreases. Furthermore, it can be noted that, for a particular value of $\Gamma$, the best case scenario inter-infrastructure failure probability is the lowest, while the worst case inter-infrastructure probabilities of failure is the highest. Though the inter-infrastructure failure probabilities are summarized here, for $I_3$ there exists two different columns $I_1 \rightarrow I_3$ and $I_2 \rightarrow I_3$. Using the relative importance matrix, the inter-infrastructure failure probability is calculated for $I_3$ as depicted in Algorithm~\ref{algo:cv}. In this study, three relative importance matrices are considered. As $I_1$ is only dependent on $I_2$, and $I_2$ is only dependent on $I_1$, the entries corresponding to $I_1$ and $I_2$ in all the three matrices are $1$. According to the construction, all the other entries for the columns corresponding to $I_1$ and $I_2$ are $0$. In $R_1$, it is considered that $I_3$ is equally dependent on $I_1$ and $I_2$. Hence, the importance of $I_1$ and $I_2$ on $I_3$ are $0.5$ each. However, in $R_2$, it is assumed that $I_1$ is less important for $I_3$ compared to $I_2$; and in $R_3$, it is assumed that $I_1$ is more important for $I_3$ compared to $I_2$. While constructing $R_2$ and $R_3$, it should be noted that the column sum of the $I_3$ has to be $1$.  

\begin{figure*}[ht]
\centering
    \includegraphics[ width=\textwidth]{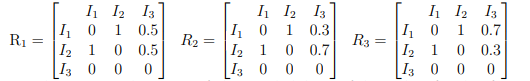}
\end{figure*}

 As depicted in Fig.~\ref{AllR11}, the average failure probabilities of the nodes of each infrastructure versus $\Gamma$, considering different relative importance matrices $R_1$, $R_2$, and $R_3$ are plotted. As $I_1$ is only dependent on $I_2$, and $I_2$ is only dependent on $I_1$, the entries in the relative importance matrix are the same for $R_1$, $R_2$, and $R_3$ for the columns corresponding to $I_1$ and $I_2$. Hence in Fig.~\ref{AllR11}(a), Fig.~\ref{AllR11}(b), and Fig.~\ref{AllR11}(c), there exists one curve each for $I_1$ and $I_2$. However, as $I_3$ is dependent on both $I_1$ and $I_2$, the entries in the $R_1$, $R_2$, and $R_3$ are different and in each of Fig.~\ref{AllR11}(a), Fig.~\ref{AllR11}(b), and Fig.~\ref{AllR11}(c) we have three curves for $I_3$ each corresponding to each relative importance matrix. There are two important observations from these plots which are as follows: 1) the failure probabilities increase as $\Gamma$ increase for the best case scenario as in Fig.~\ref{AllR11}(a) and the failure probabilities decrease as $\Gamma$ increase for the worst case scenario in Fig.~\ref{AllR11}(c) This supports the proposition~\ref{prop1}; and, 2) for $I_3$ as we have different curves for the different relative importance matrices, it is observed that the probability of failure is highest for $R_2$. This is because, in $R_2$, the importance of $I_2$ is more and the inherent failure probability of $I_2$ is more compared to $I_1$. Hence, the induced vulnerability to $I_3$ is higher for $R_2$ compared to $R_1$ and $R_3$.  

\begin{table}
\centering
\caption{Average inter-infrastructure probabilities of failure of one infrastructure induced by another versus $\Gamma$ for different scenarios after $1$ iteration.}
\label{inter2}
\begin{adjustbox}{max width=\linewidth,center}
\begin{tabular}{|l|l|l|l|l|l|}
\hline
\textbf{$\Gamma$} & \textbf{Scenario} & \textbf{$I_1 \rightarrow I_2$} & \textbf{$I_2 \rightarrow I_1$} & \textbf{$I_1 \rightarrow I_3$} & \textbf{$I_2 \rightarrow I_3$} \\ \hline
\multirow{3}*{\makecell{0.3}}
            & Best case &  $1.01*10^{-6}$ &  0.12 & $2.3*10^{-6}$ & 0.13\\
                        \cline{2-6}
           & Worst case& 0.11 & 0.40 & 0.099 & 0.35\\
                        \cline{2-6}
           & Average case& 0.04 & 0.38 & 0.043 & 0.33 \\
\hline \hline 
\multirow{3}*{\makecell{0.5}}
            & Best case &  $9.2*10^{-4}$ &  0.26 & $8.8*10^{-4}$ & 0.25\\
                        \cline{2-6}
           & Worst case& 0.08 & 0.39 & 0.073 & 0.33\\
                        \cline{2-6}
           & Average case& 0.046 & 0.37 & 0.044 & 0.32 \\
\hline \hline
\multirow{3}*{\makecell{0.7}}
            & Best case &  $9*10^{-3}$ &  0.35 & $9*10^{-4}$ & 0.31\\
                        \cline{2-6}
           & Worst case& 0.05 & 0.36 & 0.035 & 0.32\\
                        \cline{2-6}
           & Average case& 0.037 & 0.36 & 0.032 & 0.32 \\
\hline 
\end{tabular}
\end{adjustbox}
\end{table}

\begin{figure*}[ht]
\centering
    \includegraphics[ width=\textwidth]{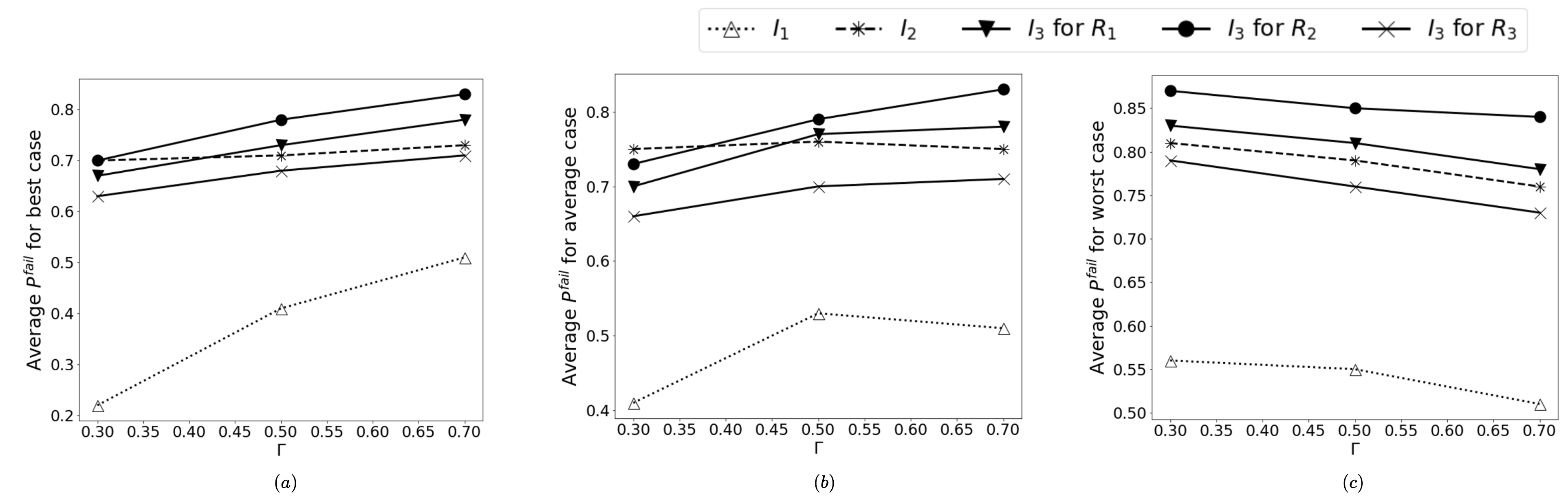}
    \caption{The probabilities of failure for each infrastructure versus $\Gamma$ for different relative importance matrices under consideration for the scenarios: (a) best case, (b) average case, and (c) worst case.}
    \label{AllR11}
\end{figure*}

\subsection{Comprehensive vulnerability estimation}
Before proceeding with the final step of estimating the comprehensive vulnerability of the interdependent infrastructure systems, first we hold some of the parameters used in our study constant, depicting our base case, to understand the failure propagation dynamics in the interdependent infrastructure system over time. Considering the mean initial failure probabilities to be same as the base case with ($\lambda_{base}$), $\Gamma = 0.5$, and the relative importance matrix as $R_1$, the failure probabilities of each node in time $[0,24]$ considering both the intra-infrastructure and inter-infrastructure vulnerabilities are estimated. First, the failure probabilities for every node of the infrastructures under the three different scenarios viz. the worst case, average case and the best case corresponding to the different number of iterations are estimated. Then corresponding to each scenario and the number of iterations, the average probabilities of failure over the nodes of a particular infrastructure are obtained. That is, the average probability of failure of infrastructure $I_i$ is, $\frac{\sum_{j=1}^{N_i} P_{ij}^{fail}}{N_i}$ where, $N_i$ is the number of nodes in infrastructure $I_i$. In Fig.~\ref{sim_b}, the probabilities of failure in time $[0,24]$ for the infrastructures are depicted for the best case scenario as a function of the number of days (i.e., iterations). It is observed that as the number of iterations (days) increase, the probabilities of failure increase. However, the rate of increase decreases with increase in the number of iterations. In Figs.~\ref{sim_w} and~\ref{sim_a} respectively, the failure probabilities of the infrastructures for the worst case and average case scenarios are depicted. Though there exist similar patterns for the different cases, the probabilities of failure for the worst case scenario is the highest, followed by the probabilities of failure for the best case. In fact, the probabilities of failure for the best case is observed to be the lowest, as expected.

\begin{figure*}[h]
    \centering
    \begin{subfigure}[t]{0.31\textwidth}
        \centering
        \includegraphics[width=\textwidth]{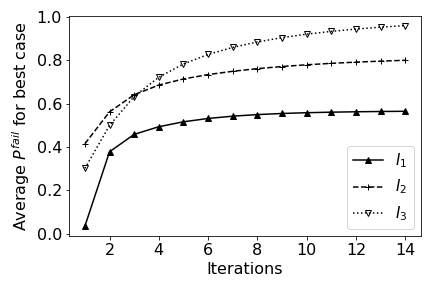}
        \captionsetup{justification=centering}
        \caption{}
        \label{sim_b}
    \end{subfigure}
    \hfill
    \begin{subfigure}[t]{0.31\textwidth}
        \centering
        \includegraphics[width=\textwidth]{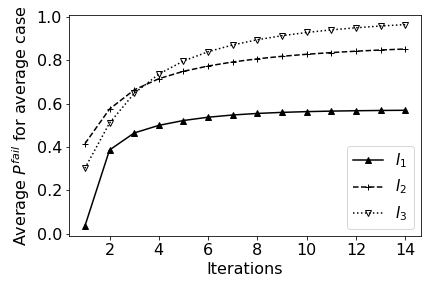}
        \captionsetup{justification=centering}
        \caption{}
        \label{sim_a}
    \end{subfigure}
    \hfill
        \begin{subfigure}[t]{0.31\textwidth}
        \centering
        \includegraphics[width=\textwidth]{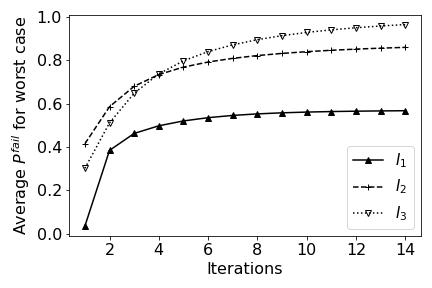}
        \captionsetup{justification=centering}
        \caption{}
        \label{sim_w}
    \end{subfigure}
    \caption{Number of iterations versus the average probability of failure of the nodes for each infrastructure in time $t \in [0,24]$, for (a) best case scenario, (b) average case scenario and (c) worst case scenario where $\Gamma = 0.5$ and the relative importance matrix is $R1$.}
\label{in}
\end{figure*} 

The key findings from the results obtained using our framework are summarized as follows. 

\begin{itemize}
    \item The \textbf{intra-infrastructure} vulnerability is observed to be inversely proportional to the number of redundancies inbuilt in the infrastructure system. This indicates that allocating resources to add redundancies in an existing infrastructure system is essential to reduce its risk of failure.
    \item For the same value of $\Gamma$ and the relative importance matrix, it is observed that the failure probabilities due to \textbf{inter-infrastructure} cascade under the worst case scenario is more than that of the average case scenario, which in turn is higher than the best case scenario. 
    The worst case scenario is realized when there are more critical nodes (less redundant nodes), failure of which induces failure of nodes of the child infrastructure and well suited for risk averse conservative design. On the other hand, the best case is realized when there are redundant parent infrastructure nodes and depict the opportunistic or risk seeking design of the system. 
    \item The parameter $\Gamma$ is inversely proportional to the geographical proximity or neighborhood of the components of two infrastructure systems under consideration. As the threshold $\Gamma$ increases, the neighborhood size decreases. Hence, for a particular infrastructure and particular relative importance matrix, the best case scenario failure probability increases as $\Gamma$ increases (geographical neighborhood decreases) due to decrease in redundancy. On the other hand, for the worst case scenario, the failure probability decreases as $\Gamma$ increases (geographical neighborhood decreases), due to decrease in the number of critical components of parent infrastructure which may lead to failure of the child component. 
    \item The vulnerability of an infrastructure is dependent on the importance of each parent infrastructure systems. If a particular infrastructure is critically dependent on another highly vulnerable infrastructure, then the vulnerability of the dependent infrastructure also increases. This indicates that working independently in an isolated way, i.e., not communicating with the managers of the related interdependent infrastructure systems would underestimate the risk of failure. Thus, working collaboratively and sharing the necessary information, aiming towards minimizing the comprehensive vulnerability of the interdependent infrastructure systems, is of utmost importance for enhancing the overall resilience of the infrastructure systems.
\end{itemize}

To better communicate the comprehensive vulnerability of the infrastructure components, a spatial distribution of the infrastructure systems' vulnerability under the best and the worst case scenarios is presented. Leveraging the Voronoi partitioning approach the vulnerability map is presented. In the Voronoi partition of the space, for each point or seed (node of the infrastructure) there is an associated region known as a Voronoi cell consisting of all the points of the plane that are closest to that particular seed than to any other seeds~\cite{voronoi}. For each node of an infrastructure, the corresponding Voronoi cell can be thought of the area that is served by that particular node, as all the points within the Voronoi cell are closest to that particular node compared to the other nodes of the infrastructure. The Voronoi partition of the region is constructed to visualize the differential vulnerabilities of the different regions measured by the failure probabilities of each component within a Voronoi cell in time $[0,24]$, due to the infrastructure component failure for the best case and the worst case scenarios. Fig.~\ref{vor_eleca} depicts the failure probabilities of the electric buses under the best case scenario, whereas Fig.~\ref{vor_elecb} presents the comprehensive vulnerability map of the electricity infrastructure for the worst case scenario. As the water is distributed only from the distribution reservoirs, the Voronoi partition consisting of only the distribution reservoirs for the best case and worst case scenarios are respectively depicted in Figs.~\ref{vor_watb}  and~\ref{vor_watw}. Each Voronoi cell in Figs.~\ref{vor_watb} and~\ref{vor_watw} respectively depict the probability of the region to be without water supply due to the disruption owing to the intra-infrastructure and inter-infrastructure cascading failures. Similarly, Figs.~\ref{vor_scn_b} and~\ref{vor_scn_w} depict the comprehensive vulnerability of the region due to failure of the supply chain network. It is observed that on average, the worst case scenario failure probabilities are higher compared to that of the best case scenario as expected. Such a vulnerability map will be particularly helpful for the infrastructure managers for identifying which geographical area of service providing sector is more vulnerable and require corrective measures and protection plans to reduce such vulnerability. 

\begin{figure*}[h]
    \centering
    \begin{subfigure}[t]{0.31\textwidth}
        \centering
        \includegraphics[width=\textwidth]{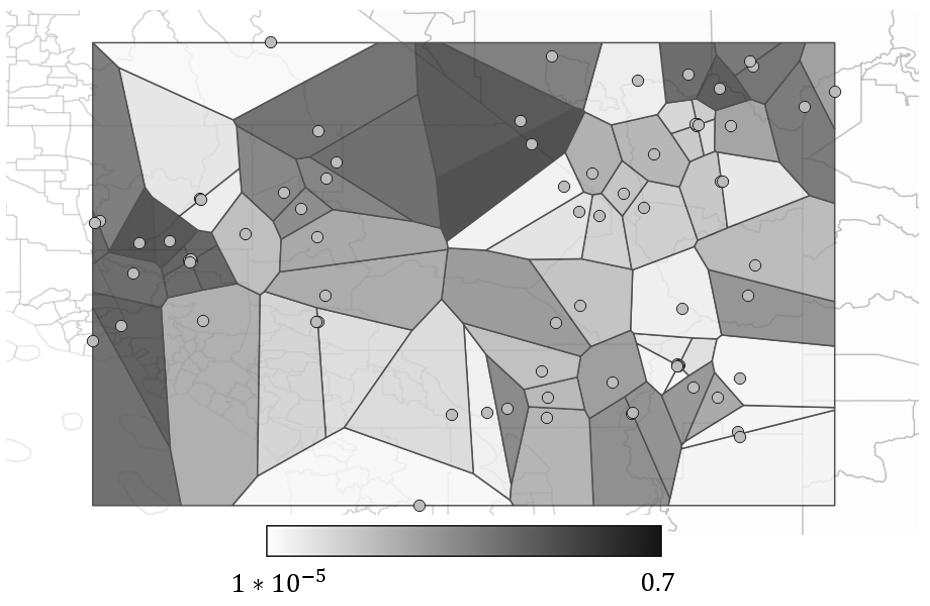}
        \captionsetup{justification=centering}
        \caption{}
        \label{vor_eleca}
    \end{subfigure}
    \hfill
    \begin{subfigure}[t]{0.31\textwidth}
        \centering
        \includegraphics[width=\textwidth]{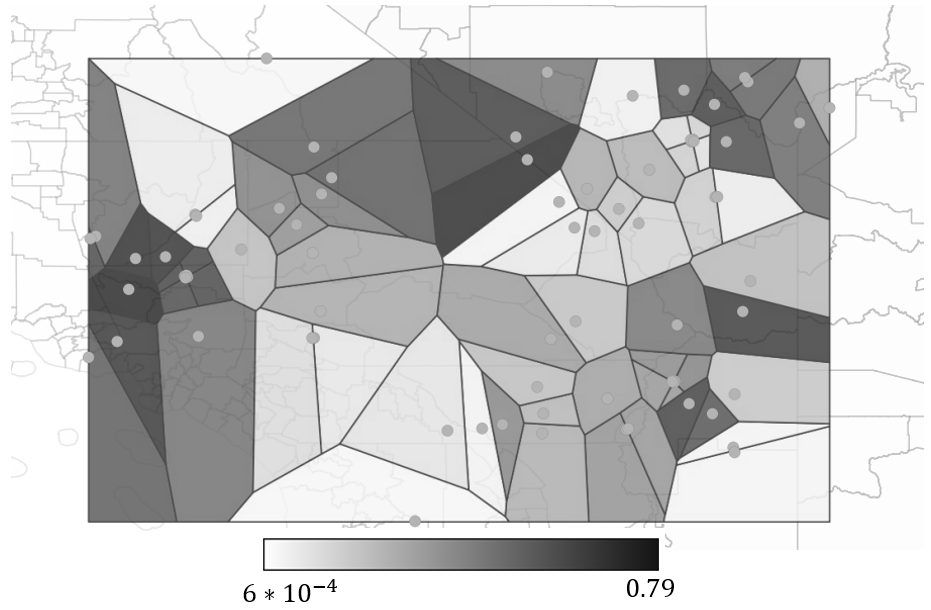}
        \captionsetup{justification=centering}
        \caption{}
        \label{vor_elecb}
    \end{subfigure}
    \hfill
    \begin{subfigure}[t]{0.31\textwidth}
        \centering
        \includegraphics[width=\textwidth]{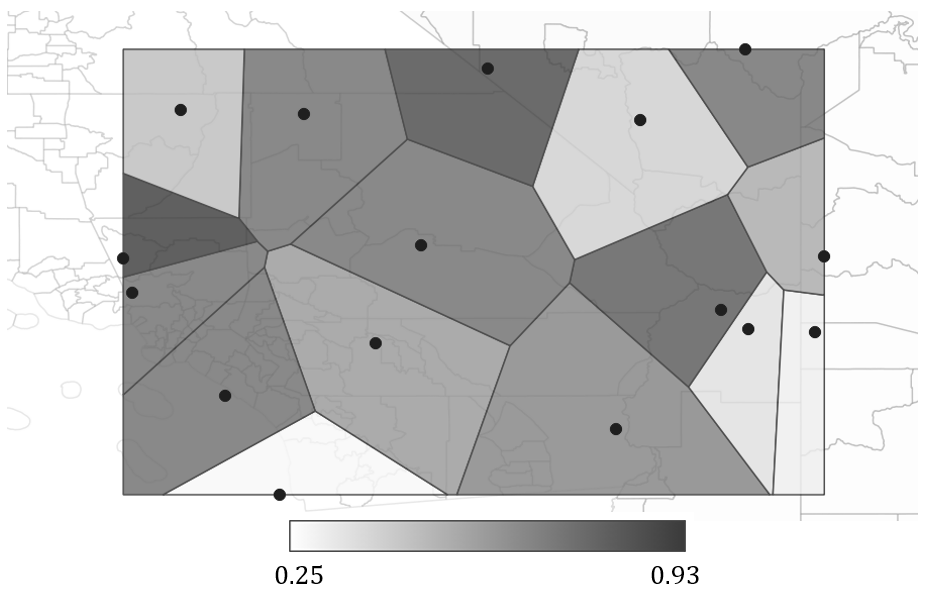}
        \captionsetup{justification=centering}
        \caption{}
        \label{vor_watb}
    \end{subfigure}
    \hfill
    \begin{subfigure}[t]{0.31\textwidth}
        \centering
        \includegraphics[width=\textwidth]{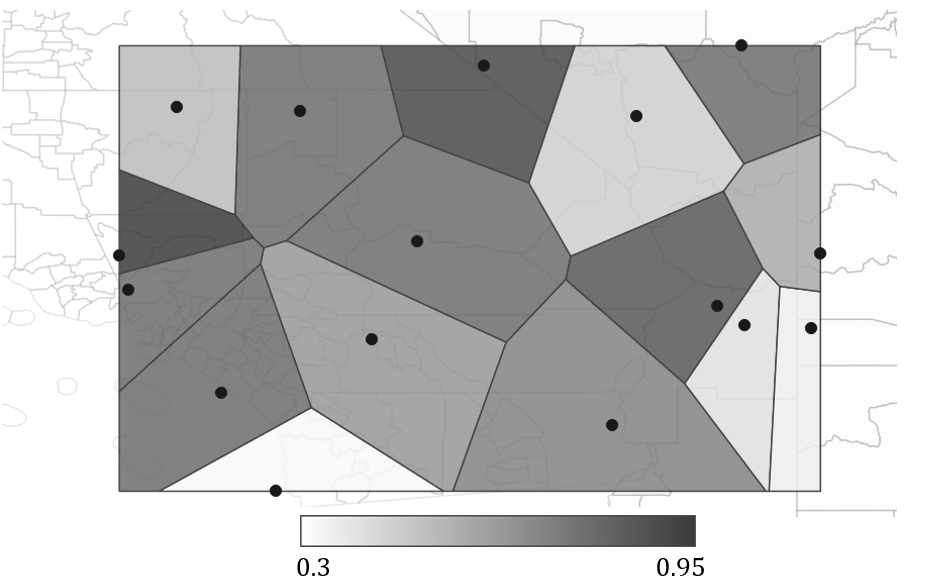}
        \captionsetup{justification=centering}
        \caption{}
        \label{vor_watw}
    \end{subfigure}
    \hfill
    \begin{subfigure}[t]{0.31\textwidth}
        \centering
        \includegraphics[width=\textwidth]{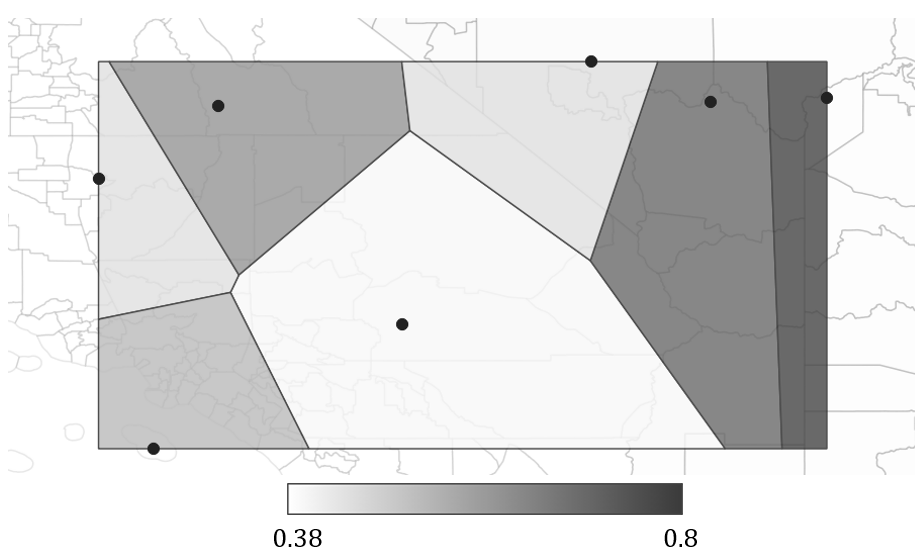}
        \captionsetup{justification=centering}
        \caption{}
        \label{vor_scn_b}
    \end{subfigure}
    \hfill
    \begin{subfigure}[t]{0.31\textwidth}
        \centering
        \includegraphics[width=\textwidth]{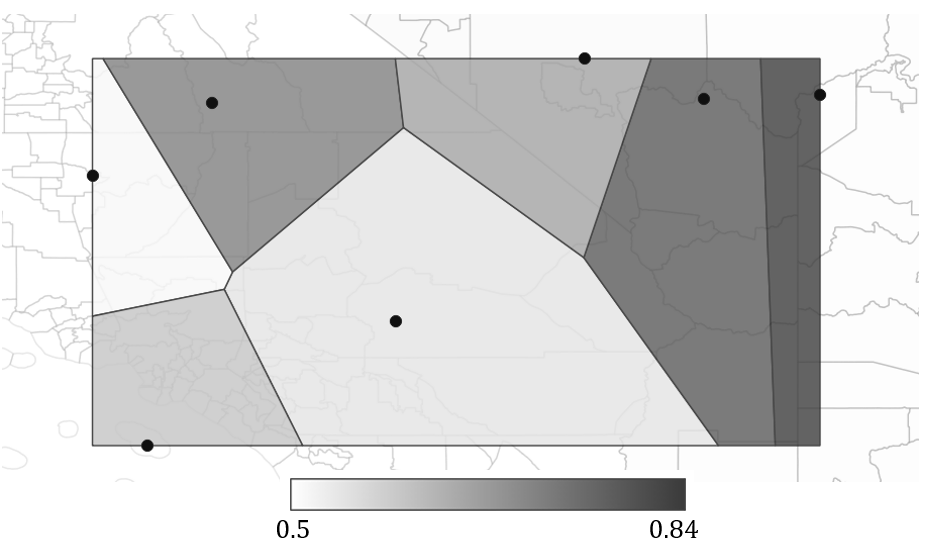}
        \captionsetup{justification=centering}
        \caption{}
        \label{vor_scn_w}
    \end{subfigure}
    \caption{The spatial distribution of failure probabilities of the infrastructure nodes in time $[0,24]$ after two iterations considering $R1$ as the relative importance matrix and $\Gamma = 0.5$ for the (a) best case scenario for electricity, (b) worst case scenario for electricity infrastructure, (c) best case scenario for water distribution network, (d) worst case scenario for water distribution network, (e) best case scenario for the supply chain network, and (f) worst case scenario for the supply chain network. The failure probabilities are deployed over a Voronoi partition of the region.}
\label{in}
\end{figure*}

\section{CONCLUSION}
\label{conc}
In this research, we present a novel framework using a hybrid of deterministic and simulation-based approaches to estimate the vulnerability of multiple interdependent infrastructure system arising from the cascading failure propagation over time, within an infrastructure network and failures induced due to the interdependencies with other infrastructure systems. Leveraging a generalized algorithm, a dynamic fault tree and continuous time Bayesian network are developed to model the failure propagation within an infrastructure. Thereafter, using a simulation-based approach a network of interdependent infrastructure networks is constructed to capture the vulnerability of one infrastructure induced by the vulnerabilities of other interdependent infrastructure. The compound vulnerability is estimated for three different scenarios (best, worst and average case in terms of likelihood of the system failure), and is presented as a single metric on a dynamic geo-map for better communication to the stakeholders. The framework is implemented for a case study using synthetic data on the interdependent infrastructure systems like electricity distribution, water distribution and a supply chain network. Several user-defined parameters including the failure rates of the infrastructure components and the relative importance of the infrastructure interdependencies are used in the framework that makes the framework flexible and generalized to be applicable to any networked infrastructure systems. Although, the parameters depicted in the case study do not represent the exact values from a real life scenario, in presence of real data these parameters may be altered without any alteration of the overall framework. The scenario based supply-demand characteristic viz., worst case, average case or best case are chosen by the infrastructure managers according to their risk perception. However, in real life, a combination of these three scenarios may be realized, where some components are critically dependent on another (worst case) while other components have some amount of redundancy (best/average case). Our framework can be easily extended to tackle this situation. Besides, this work can be extended to any number of interdependent infrastructure system under incomplete information about the supply demand characteristics of the infrastructures.






\newpage
\printbibliography

\newpage
\section*{Appendix}
\label{app}
\subsection*{Derivation of the closed form solutions for the probabilities of failure of DFT gates}
\label{app1}
Using the unit step function defined by,
\[
u(t-\tau) = \begin{cases}
0, \quad \text{if $t < \tau$}\\
\frac{1}{2}, \quad \text{if $t = \tau$}\\
1, \quad \text{if $t>\tau$}
\end{cases}
\]
is used in the CTBN to denote the event occuring at time $t$ and not before that. Furthermore, the impulse function or the Dirac delta function defined by,
\[
\delta(t-\tau) = 0, \quad \text{if $t \neq \tau$} \\ \quad
\text{and} \quad
\int_{-\infty}^\infty \delta(t-\tau) dt = 1
\]
is used for modeling a variable taking a specific unique value. Now, we describe the probability density function (PDF) of the $AND$, $OR$, and the $WSP$ gate of the DFT. 

\textbf{AND gate: } If the nodes $A$ and $B$ with the marginal PDF $f_A(a)$ and $f_B(b)$ respectively are connected by $AND$ gate to produce the output $X$, then the conditional PDF of $f_{X|A,B}(x|a,b)$ is given by,
\[
f_{X|A,B}(x|a,b) = u(b-a)\delta(x-b) + u(a-b)\delta(x-a)
\]
where the first term denotes when $A$ fails before $B$, the state of $X$ is same as the state of $B$ which failed later; and the second term denotes when $B$ fails before $A$, then the state of $X$ is same as the state of $A$. The marginal probability density of $X$ is obtained by marginalizing the joint distribution $f_{ABX}(a,b,x)$, as
\[
f_X(x) = \int_{-\infty}^\infty \int_{-\infty}^\infty f_{ABX}(a,b,x) db da = \int_{-\infty}^\infty \int_{-\infty}^\infty f_{X|AB}(x|a,b) f_B(b) f_A(a) db da = [F_B(b)F_A(a)]'
\]
Let us consider that the probability density function of failure time of $A$ follows exponential distribution with rate $\lambda_A$ and for $B$, the probability density function of failure time is exponentially distributed with rate $\lambda_B$. Hence, the probability of failure of $X$ in time $[0,t]$ is given by,
\[
F_X(t) = \int_0^t f_X(x)dx = F_B(t)F_A(t) = 1- e^{-\lambda_At} - e^{-\lambda_Bt} + e^{-(\lambda_A+\lambda_B)t}
\]

\textbf{OR gate: } If the nodes $A$ and $B$ with the marginal PDF $f_A(a)$ and $f_B(b)$ respectively are connected by $OR$ gate to produce the output $X$, then the conditional PDF of $f_{X|A,B}(x|a,b)$ is given by,

\[
f_{X|A,B}(x|a,b) = u(b-a)\delta(x-a) + u(a-b)\delta(x-b)
\]
where, the first term denotes if $A$ fails before $B$, then the state of $X$ is same as the state of $A$ which fails first; on the other hand, if $B$ fails before $A$, then the state of $X$ is same as the state of $B$. Using the similar procedure of the $AND$ gate, we have, 

\[
f_X(x) = \int_{0}^{\infty} \int_{0}^{\infty} f_{ABX}dbda = f_A(x) + f_B(x) - [F_B(b)F_A(a)]'.
\]

Finally, considering for $A$ and $B$, the time of failure follows exponential distribution with rate $\lambda_A$ and $\lambda_B$ respectively, the probability of failure of $X$ in time $[0,t]$ is given by,

\[
F_X(t) = \int_{0}^{t}  f_{X}(x)dx = F_A(t) + F_B(t) - F_B(t)F_A(t) = 1- e^{-(\lambda_A+\lambda_B) t}
\]

\textbf{WSP gate: } 
In a two input WSP gate, say, $A$ is the primary unit and $B$ is the spare unit. When the system starts, the component $A$ stats working and the component $B$ is in standby or dormant mode. In dormant mode, the failure rate is reduced by a factor $\alpha$. First to model the failure of the node $B$, we have,

\[
f_{B|A}(b|a) = u(a-b)\alpha f_{B_i}(b)[1-F_{B_i}(b)]^{\alpha-1}+ u(b-a)f_{B_i}(b-a)[1-F_{B_i}(a)]^\alpha
\]

where, $f_{B_i}$ and $F_{B_i}$ are the in isolation 
Considering the exponential time of failure for $A$ and $B_i$ in isolation with rates $\lambda_A$ and $\lambda_B$ respectively, the probability density function of the node $B$ is given as,

\begin{align*}
   & f_B(b) = \int_0^\infty f_{B|A}(b|a) f_A(a) da, \\
   & = \alpha f_{B_i}(b)[1-F_{B_i}(b)]^{\alpha -1}[1-F_A(b)] + \int_0^b f_{B_i}(b-a)[1-F_{B_i}(a)]^\alpha f_A(a) da \\
  &  = \alpha \lambda_B e^{-b(\lambda_A +\lambda_B \alpha)} + \frac{\lambda_A \lambda_B}{\lambda_B- \lambda_B \alpha - \lambda_A} [e ^{-b \lambda_A - b \alpha \lambda_B} - e^{-\lambda_B b}]
\end{align*}

Hence, the probability of failure of $B$ in $[0,t]$ is given by, 
\begin{align*}
   & F_B(t) = \int_0^t f_B(b)db, \\
  & = \int_0^t \alpha \lambda_B e^{-b(\lambda_A +\lambda_B \alpha)} db + \int_0^t \frac{\lambda_A \lambda_B}{\lambda_B- \lambda_B \alpha - \lambda_A} [e ^{-b \lambda_A - b \alpha \lambda_B} - e^{-\lambda_B b}] db \\
  & = \frac{\alpha \lambda_B (e^{-t(\lambda_A + \alpha \lambda_B)}-1)}{-\lambda_A - \alpha \lambda_B} - \frac{\lambda_A(-(e^{-\lambda_B t}-1)(\lambda_A + \alpha \lambda_B) + \lambda_B e^{-t(\lambda_A + \alpha \lambda_B)}- \lambda_B)}{(\lambda_A + \alpha \lambda_B) (\lambda_B - \lambda_A -\alpha \lambda_B)}
\end{align*}

The output $X$ of the WSP is an $AND$ gate, i.e., $X = A \,AND\, B$. 

Hence, the probability of failure of $X$ in time $[0,t]$ is given by,
\begin{equation*}
    F_X(t) = (1-  e^{-\lambda_A t})(\frac{\alpha \lambda_B (e^{-t(\lambda_A + \alpha \lambda_B)}-1)}{-\lambda_A - \alpha \lambda_B} - \frac{\lambda_A(-(e^{-\lambda_B t}-1)(\lambda_A + \alpha \lambda_B) + \lambda_B e^{-t(\lambda_A + \alpha \lambda_B)}- \lambda_B)}{(\lambda_A + \alpha \lambda_B) (\lambda_B - \lambda_A -\alpha \lambda_B)})
\end{equation*}

\textbf{Priority AND gate: } In a priority AND gate with two inputs $A$ and $B$ producing output $X$, the output fails if both the inputs fail and $A$ fails before $B$. 
Using the method proposed by Amari \textit{et al.}, the probability of failure of $X$ in $[0,t]$ is given by,
\[
F_X(t) = \int_0^t f_A(a)da \int_a^tf_B(b)db = \int_0^t f_A(a)da [\int_0^t f_B(b)db - \int_0^a f_B(b)db]
\]

Using numerical integration methods,

\[
F(t) = \sum_{i=1}^T [F_A(i\times h) - F_A((i-1)\times h)] [F_B(t) - F_B(i \times h)]
\]
where, $T$ is the number of time steps or intervals and $h= t/T$ is the step size. 

\end{document}